\documentclass[manuscript]{acmart}

\usepackage{tikz}
\usetikzlibrary{positioning}
\usepackage{algorithm}
\usepackage{cleveref}
\usepackage{todonotes}
\usepackage{balance}

\usepackage[utf8]{inputenc}
\long\def\commentbegin #1\commentend{}

\usepackage{balance}
\usepackage{algorithm}
\usepackage[noend]{algpseudocode}
\usepackage{ifthen}
\usepackage{comment}
\usepackage{graphicx, color, xcolor}
\usepackage{bbm}
\usepackage{booktabs}
\usepackage{thm-restate}
\usepackage{lineno}

\keywords{Gossip, Rumor Spreading, Congest, KT1, Broadcast}

 \ccsdesc[500]{Theory of computation~Distributed algorithms}
 \ccsdesc[500]{Mathematics of computing~Probabilistic algorithms}
 \ccsdesc[300]{Mathematics of computing~Discrete mathematics}

\renewcommand{\epsilon}{\varepsilon}

\renewcommand{\geq}{\geqslant}

\renewcommand{\leq}{\leqslant}

\newcommand{\Prob}[1]{\hbox{\rm I\kern-2pt P}\left[#1\right]}

\def\CONGEST{\mathcal{CONGEST}}
\def\LOCAL{\mathcal{LOCAL}}
\def\GOSSIP{\mathcal{GOSSIP}}
\def\polylog{\operatorname{polylog}}

\newcommand{\Deltath}{\Delta_{Th}}

\newcommand{\incVect}{\textbf{i}_G}
\newcommand{\sketchVect}{\Large{\textbf{s}}_G}

\newcommand{\comp}[1]{\ensuremath{\overline{#1}}}
\newcommand{\card}[1]{\mathchoice{\left\lvert#1\right\rvert}{\lvert#1\rvert}{\lvert#1\rvert}{\lvert#1\rvert}}
\DeclareMathOperator{\vol}{Vol}
\newcommand{\set}[1]{\mathchoice{\left\lbrace#1\right\rbrace}{\lbrace#1\rbrace}{\lbrace#1\rbrace}{\lbrace#1\rbrace}}

\ifdefined\ShowComment
\def\billy#1{{\color{green}\underline{\textsf{Billy:}}} {\color{blue} \emph{#1}}}
\def\gopal#1{{\color{red}\underline{\textsf{Gopal:}}} {\color{blue} \emph{#1}}}
\def\fabien#1{{\color{orange}\underline{\textsf{Fabien:}}} {\color{blue} \emph{#1}}}
\else
\def\billy#1{}
\def\gopal#1{}
\def\fabien#1{}
\fi

\title{Fast Gossip-Based Rumor Spreading Using Small Messages}

\author{Fabien Dufoulon} 
\orcid{0000-0003-2977-4109}
\affiliation{%
\institution{Lancaster University}
\city{Lancaster}
\country{UK}
}
\email{f.dufoulon@lancaster.ac.uk}

\author{William K. {Moses Jr.}}
\orcid{0000-0002-4533-7593}
\affiliation{%
\institution{Durham University}
\city{Durham}
\country{UK}
}
\email{wkmjr3@gmail.com}

\author{Gopal Pandurangan}
\orcid{0000-0001-5833-6592}
\affiliation{%
  \institution{University of Houston}
  \city{Houston}
  \country{USA}
}
\email{gopal@cs.uh.edu}

\copyrightyear{2026}
\acmYear{2026}
\setcopyright{cc}
\setcctype{by}
\acmConference[PODC '26]{ACM Symposium on Principles of Distributed Computing}{July 06--10, 2026}{Egham, United Kingdom}
\acmBooktitle{ACM Symposium on Principles of Distributed Computing (PODC '26), July 06--10, 2026, Egham, United Kingdom}
\acmDOI{10.1145/3796701.3815943}
\acmISBN{979-8-4007-2512-8/2026/07}

\begin{document}

\begin{abstract}

We study gossip algorithms for the fundamental rumor spreading problem, where the goal is to disseminate a rumor from a given source node to all nodes in an arbitrary (and unknown) graph. Gossip algorithms allow each node to call only  \textit{one} neighbor per round 
and are therefore highly message-efficient, with low per-node communication overhead per round. The state of the art works present fast gossip algorithms, however they typically leverage \textit{large-sized messages}. This undermines the light-weight communication advantage of gossip, since even though only one neighbor is contacted per round, the message size can be linear in $n$, the network size. Hence, a fundamental question is whether one can perform fast gossip using \textit{small}  messages.

 The main contribution of this paper is to answer the above question in the affirmative and present two gossip algorithms that achieve fast rumor spreading \textit{using messages of $\polylog{n}$ size}.  Specifically, we show the following results:
 \begin{enumerate}
 \item We present a gossip algorithm for rumor spreading that runs in $O(c \log n / \Phi_c)$ rounds for every $c \geq 1$, and $\Phi_c$ is the weak conductance. Our algorithm's run time \textit{not only improves} over the
Censor-Hillel-Shachnai bound [SODA 2011; SICOMP 2012], but more significantly, it uses \textit{messages of small ($\polylog{n}$) size}, unlike the prior work that used messages of large (at least linear in $n$) size. Our bound in terms of weak conductance is essentially optimal.
\item We also present a gossip algorithm for rumor spreading that depends  on the network diameter (and is independent of the graph's conductance), which runs in $\tilde{O}(D+\sqrt{n})$ rounds with high probability and uses small ($\polylog{n}$ size) messages. Our bound is a significant improvement over the gossip algorithm of 
$\tilde{O}(\sqrt{nD})$ due to Ghaffari and Kuhn [DISC 2018], which also uses small-sized messages. We note that our algorithm (unlike that of Ghaffari and Kuhn) is optimal for when $D = \Omega(\sqrt{n})$. Furthermore, our gossip algorithm can be modified to output a minimum spanning tree (MST) in the same number of rounds, which is essentially round-optimal (even for non-gossip algorithms).
\end{enumerate}
Our gossip algorithms use \textit{graph sketches} [Ahn, Guha, McGregor, SODA 2012] in a novel way to overcome communication bottlenecks and achieve small communication overhead with small message sizes.

\end{abstract}

\maketitle

\section{Introduction}
\label{sec:intro}
\textit{Rumor spreading} (also called \textit{broadcasting} or \textit{information dissemination}) is a fundamental network communication primitive that has been studied extensively 
for many decades under various models and assumptions~\cite{CHKM12,CHKM12-journal,CS11,CS12,CGLP18,FPRU90,FG85,gia1,vertex-exp,GSS14,gia2,H13,H15,P87}. In the rumor spreading problem, we are given an arbitrary  (connected)  communication graph $G = (V,E)$ and a source node $v \in V$ with a piece of
information (a.k.a. ``rumor") and the goal is to disseminate the rumor to all other nodes in the graph. As discussed later, it is easy to modify rumor spreading algorithms to compute node \textit{aggregate} functions such as minimum, maximum,  sum or average. (Throughout we use $n  = |V|$ to denote the
number of nodes in $G$, $m = |E|$ for the number of edges in the network, and $D$ for the graph diameter.) 

We study \textit{gossip} algorithms for rumor spreading. In gossip,  each node is allowed to call only  \textit{one} neighbor per round.
Hence, it is very message-efficient, with a small communication overhead per node per round. This is in contrast to other message-heavy broadcasting techniques such as \textit{flooding}, where a node  (with the rumor) forwards it to \textit{all} neighbors (which can be as large as $\Theta(n)$) in a round.
It is well-known that there are graphs where flooding can take $\Theta(m)$ messages to spread the rumor to all nodes, but gossip takes only $\tilde{O}(n)$ messages ($\tilde{O}$ notation hides logarithmic in $n$ factors).\footnote{For example, in dense graphs with $\Theta(n^2)$ edges and having constant conductance, uniform gossip  takes $O(\log n)$ rounds and, hence, $O(n \log n)$ messages, whereas flooding takes $\Theta(n^2)$ messages.} 

A simple and well-studied gossip algorithm is \textit{uniform gossip} where each node calls a \textit{random} neighbor and exchanges the rumor (if any of them have it).  More precisely, in uniform gossip, also called (uniform) \textit{push-pull} gossip, 
in each round, each  node  $v$ with the rumor contacts a random neighbor  and \textit{pushes} the rumor to it,
and each node $u$ that does not have the rumor, contacts a random neighbor and \textit{pulls} the rumor from it (if the neighbor has it).
A long line of research~\cite{CLP10-stoc,CLP10-soda,CLP11,E06,KSSV00} on uniform gossip culminated in showing that it takes $O(\frac{\log n}{\Phi})$ rounds 
(w.h.p.)\footnote{W.h.p. stands for with high probability and denotes $1 - 1/n^d$, for an arbitrary $d \geq 1$.}
where $\Phi$ is the \textit{conductance} (see Section~\ref{prelims:weak-conductance})  of the graph \cite{CGLP18,Giakkoupis2011}.  In general, this bound is also tight as there are graphs where rumor spreading needs at least $\Omega(\frac{\log n}{\Phi})$ rounds (e.g., a constant degree expander) \cite{CGLP18,Giakkoupis2011}. 

Informally, conductance $\Phi \in [0,1]$ measures how well the graph is connected, and is high for well-connected graphs such as cliques and expanders but small for graphs that have bottlenecks (such as the dumbbell graph mentioned below).  
 Uniform gossip is slow if the conductance of the network is small, which happens when the graph has bottlenecks. In such a case,
uniform gossip can take as much as $\Theta(n \log n)$ rounds.\footnote{Note that even a weaker version of gossip called \textit{push-gossip} where only nodes (with the rumor) push, achieves rumor spreading in $O(n \log n)$ rounds \cite{FPRU90}.}
For example, consider a dumbbell graph with two cliques each of size $n/2$ connected by an edge; here uniform gossip will take $\Theta(n)$ rounds in expectation and $O(n \log n)$ rounds w.h.p.,
even though the graph diameter is only 3.

To mitigate the slowness of uniform gossip in graphs of low conductance, a series of papers studied \textit{non-uniform} gossip that could overcome communication bottlenecks. 
In non-uniform gossip, it is \textit{not} required that a node contacts a \textit{random} neighbor in a round; rather, any (one)  neighbor can be contacted.
The first such significant result was a non-uniform gossip algorithm by Censor-Hillel and Shachnai~\cite{CS11,CS12}  that achieved a running time that depended on so-called  \textit{weak conductance} 
of the graph, which can be substantially larger than its conductance.  The weak conductance of a graph, denoted by $\Phi_c$, is parameterized by $c \geq 1$, and measures the conductance restricted to node subsets of size up to $n/c$
   (see Section \ref{prelims:weak-conductance} for a formal definition). For example, the conductance (which is the same
   as weak conductance with $c=1$) of the above dumbbell graph is  small --- $\Theta(1/n^2)$ --- but its weak conductance
$\Phi_c$ is $\Theta(1)$ for $c = 2$. 
The Censor-Hillel-Shachnai algorithm achieves rumor spreading in $O(\frac{c\log n}{\Phi_c} +c^2)$ rounds  (with high probability) for any $c>1$. Hence, for the dumbbell graph, this algorithm takes $O(\log n)$ rounds, which is exponentially faster than uniform gossip.
An important point to note about this algorithm is that the messages sent can be of very \textit{large size} (up to linear in $n$ due to the fact that a message can consist of the identities of a linear number of nodes) and hence the message overhead is very large even to spread a \textit{single rumor}.

 A subsequent work by Censor-Hillel, Haeupler, Kelner, and Maymounkov~\cite{CHKM12}  
 presented a non-uniform gossip algorithm that runs in time that is \textit{independent} of the conductance of the graph. 
 This algorithm ran in $O(D+\polylog{n})$ rounds (with high probability), where $D$ is the graph diameter. This is near optimal since $\Omega(D)$ is a trivial lower bound even for the more powerful flooding algorithm (where a node can send messages to all its neighbors in a round). 
 Thus, it was surprising that gossip, which is constrained to contact only one neighbor, can almost (up to an additive $\polylog{n}$ term) match the bound obtained by flooding. Hence, this work established that gossip is almost as powerful as flooding.
 Later, Haeupler~\cite{H13,H15} showed that the diameter running time bound can even be achieved \textit{deterministically}, by presenting
 a deterministic non-uniform gossip algorithm that ran in $2(D\log n + \log^2 n)$ rounds.  \textit{It is important to note that all the above gossip algorithms 
 also use messages of large size (at least linear in $n$).}

To summarize,  while the above gossip algorithms can be significantly faster than uniform gossip, a main drawback of the algorithm is that they use \textit{large-sized messages}  to spread even a \textit{single} rumor. (In contrast, in uniform gossip, only the rumor is exchanged, and hence, the message size is small.) This undermines the light-weight communication advantage of gossip, since even though only one neighbor is contacted per round, the message size can be linear in $n$, the network size. Hence, a fundamental question is:

\begin{center}
    \fbox
    {
        \begin{minipage}{46em}
           Can we achieve fast rumor spreading (in terms of weak conductance or just in terms of the diameter as in the results discussed above) using only \textit{small (i.e., $\polylog{n}$ size)} messages?
 
        \end{minipage}
    }
\end{center}

 In this paper, we answer the above question in the affirmative and present two gossip algorithms that achieve fast rumor spreading \textit{using messages of $\polylog{n}$ size}. The full version of the paper~\cite{DMP26arxiv} contains missing algorithm descriptions and proofs.

\subsection{Model}
\label{subsec:intro-model}

We assume that we are given a connected \textit{arbitrary} graph $G = (V,E)$ as input. Let $|V| = n$ and $|E| = m$.

Each node has a unique ID (identifier) taken from the range $[1,N]$, where $N$ is a polynomial in $n$. 
Each ID can be represented using $O(\log n)$ bits. As is standard in prior non-uniform gossip algorithms \cite{CHKM12,CS11,CS12,H13,H15}, we assume that each node knows its own ID and the IDs of its neighbors. (This is a standard model in distributed computing, called the \textit{Knowledge-Till-Radius 1 $(KT1)$} model.) 
It is important that the network topology $G$ is unknown to the nodes. We assume all nodes have knowledge of the polynomial upper bound $N$.
Additionally, our algorithms  require extra knowledge assumptions, which are mentioned in Section~\ref{subsec:intro-our-contribs}.

The computation proceeds in \textit{synchronous} rounds. We assume a \textit{gossip-based} communication model which is very lightweight, where in a round,  a node contacts  \textit{at most one of its neighbors} in order to send and receive a message over that edge.  This setting is henceforth known as the $\GOSSIP$ model~\cite{CHKM12,CS11,CS12,H13,H15}. 
 Note that although the $\GOSSIP$ model allows each node to send a message to only one neighbor per round, one can easily simulate sending messages to any  $k \geq 1$ neighbors in a round, by performing gossip for $k$ rounds and thus blowing up the number of rounds by a factor of $k$. On the other hand, in the standard $\CONGEST$ model, a node can send a message (of small size, say $O(\log n)$ bits) to all (or a subset) of its neighbors in a round.
 In addition, in the $\LOCAL$ model, the message
size is unbounded, i.e., a node can send messages of
arbitrary size to all (or a subset) of its neighbors in a round. 

As mentioned earlier, prior non-uniform algorithms~\cite{CS11,CS12,H13,H15}
transmit messages of large (up to linear in $n$) size. Hence, they can be said to operate in the $\GOSSIP-\LOCAL$ model where each node can contact
only one neighbor in a round, but the size of
the message exchanged 
can be arbitrary. On the other hand, our algorithms 
operate in the $\GOSSIP-\CONGEST$ model, where each node can contact
only one neighbor in a round, but the size of
the message exchanged is small ($\polylog{n}$ size).\footnote{Note that one can restrict each message size to $O(\log n)$ bits instead of $O(\polylog{n})$ bits; however, this leads to (only) a $O(\polylog{n})$-factor blow up in the round complexity of our algorithms.}
Thus, a round $r$ in the $\GOSSIP-\CONGEST$ model consists of the following three steps: (i) each node contacts one neighbor, (ii) each node sends a $O(\polylog{n})$ bit message to the contacted neighbor and receives a $O(\polylog{n})$ bit reply, and (iii) after all messages in transit have been received, $u$ performs some local computation.\footnote{Note that although each node contacts only one neighbor, 
 a node can be contacted by as many neighbors as its degree in a round
 (e.g., the center node in a star graph). As standard in gossip protocols (including those that we compare here)~\cite{CLP10-stoc,CLP10-soda,CLP11,E06,KSSV00,CS11,CS12,CHKM12-journal, GK18}, we assume that the node replies to all of them in the same round.} 

\subsection{Our Contributions}
\label{subsec:intro-our-contribs}

We present two gossip algorithms that achieve fast rumor spreading \textit{using messages of $\polylog{n}$ size}, improving over prior results.  Notice that each round in a gossip algorithm results in $O(n)$ messages, so the total message complexity of a $t$ round algorithm is $O(nt)$.

\paragraph{A gossip algorithm as a function of weak conductance (see Section \ref{sec:alg})} We present the first (non-uniform)  gossip algorithm to achieve fast rumor spreading in terms of weak conductance using messages of $\polylog{n}$ size, in the $\GOSSIP-\CONGEST$ model.  
  
Our  (randomized) gossip algorithm for rumor spreading, with high probability, runs in 
$O(c \log n / \Phi_c)$ rounds for every $c \geq 1$, where $\Phi_c$ is the weak conductance (defined in Section \ref{prelims:weak-conductance}). This bound not only improves over the running time bound
of  $O(c\frac{\log n}{\Phi_c} +c^2)$ of Censor-Hillel and Shachnai~\cite{CS11,CS12}, but also uses significantly smaller messages of $\polylog{n}$ size. In other words, while the Censor-Hillel-Shachnai algorithm  runs  in the $\GOSSIP-\LOCAL$ model, our algorithm has better performance \textit{and} runs in the $\GOSSIP-\CONGEST$ model. 
Our algorithm assumes that every node knows the values of both $c$ and $\Phi_c$. 
We note that the previous best algorithm for this setting~\cite{CS11,CS12} needs this assumption in order to achieve termination. 
A comparison of our result and past work can be found in Table~\ref{table:results-conductance}.
Formally, we show the following result.

\begin{restatable}{theorem}{rumorSpreadingInWeakConductance}
\label{thm:rumorSpreadingInWeakConductance}
    There exists a $\GOSSIP-\CONGEST$ algorithm solving rumor spreading, with high probability, in $O(c \log n / \Phi_c)$ rounds with high probability, for every $c\geq1$ and where $\Phi_c$ is the weak conductance of the graph.
\end{restatable}

Additionally, one may see that our algorithm is existentially asymptotically optimal. Consider the following variant of the $c$-barbell (see Figure~\ref{fig:c-barbell}): consider $c$ expanders, each of size $n/c$ and diameter $\Theta(\log (n/c))$, connected in a line. Such a graph has weak conductance $\Phi_c = \Theta(1)$ (see Section~\ref{prelims:weak-conductance}). The diameter of such a graph is $\Theta(c \log (n/c))$ and acts as a natural lower bound for gossip. Then, we can see that our algorithm is asymptotically optimal for constant values of $c$ and almost asymptotically optimal (up to a $\log n$ factor) for larger values of $c$. Furthermore, the absence of a $c^2$ factor in our run time can result in significant time savings for some types of graphs. In particular, for larger values of $c$, i.e., $c = \Theta(n)$, the above described lower bound graph leads to a $\Theta(n)$ speedup in run time of our algorithm over that of Censor-Hillel and Shachnai~\cite{CS11,CS12}.

\begin{table*}[ht]
\footnotesize
	\caption{
	Comparison of gossip algorithms focused on weak conductance $\Phi_c$. $\tilde{O}$ notation hides logarithmic in $n$ factors, where $n$ is the number of nodes.
	} 
	\centering 
		\resizebox{0.8\columnwidth}{!}
        {%
	\begin{tabular}{|c|c|c|}
		\hline
		\textbf{Result} & \textbf{Run Time (in rounds)} & \textbf{Message Size Required (in bits)} \\
		\hline
		Censor-Hillel \& Shachnai~\cite{CS11,CS12} & $O(c\frac{\log n}{\Phi_c} +c^2)$ & $\tilde{O}(n)$ \\
		\hline
        This work - Theorem~\ref{thm:rumorSpreadingInWeakConductance} & $O(c\frac{\log n}{\Phi_c})$ & $\tilde{O}(1)$\\
        \hline
	\end{tabular}
		}
	\label{table:results-conductance}
\end{table*}

\paragraph{A gossip algorithm as a function of diameter (see Section \ref{sec:diameterAlg})}
We also present another gossip algorithm for rumor spreading in the $\GOSSIP-\CONGEST$ model that depends on the network diameter, which runs in $\tilde{O}(D+\sqrt n)$ rounds with high probability.
The randomized gossip algorithm of Censor-Hillel, Haeupler, Kelner, and Maymounkov~\cite{CHKM12} (see also the deterministic algorithm due to Haeupler~\cite{H13,H15} with a similar time bound) whose run time is a function of the graph diameter (and not conductance) ---  takes  $O(D+\polylog{n})$ rounds, but again, it uses large-sized messages, hence it works only in the $\GOSSIP-\LOCAL$ model. On the other hand, the randomized gossip algorithm of Ghaffari and Kuhn~\cite{GK18} works in the $\GOSSIP-\CONGEST$ model but takes $\tilde{O}(\sqrt{nD})$ rounds.  In this work, we give a randomized gossip algorithm that significantly improves on that latter result. Our algorithm assumes that nodes have knowledge of a constant-factor upper bound on $n$, but it does not assume any knowledge of the diameter. A comparison of our result and past work can be found in Table~\ref{table:results-diameter}.

\begin{restatable}{theorem}{rumorSpreadingGeneral}
\label{thm:rumorSpreadingGeneral}
There exists a $\GOSSIP-\CONGEST$ algorithm solving rumor spreading, with high probability, in $\tilde{O}(D + \sqrt n )$ rounds.
\end{restatable}

We note that our algorithm (unlike that of \cite{GK18}) is optimal when $D = \Omega(\sqrt{n})$. A major open question is 
whether the bound shown in Theorem \ref{thm:rumorSpreadingGeneral} is optimal.
We conjecture that this bound is indeed 
optimal (up to logarithmic factors) --- see Section~\ref{sec:conclusion}.

\begin{table*}[ht]
\footnotesize
	\caption{
	Comparison of gossip algorithms focused on diameter $D$. $\tilde{O}$ notation hides logarithmic in $n$ factors, where $n$ is the number of nodes.
	} 
	\centering 
		\resizebox{1.0\columnwidth}{!}
        {%
	\begin{tabular}{|c|c|c|}
		\hline
		\textbf{Result} & \textbf{Run Time (in rounds)} & \textbf{Message Size Required (in bits)} \\
		\hline
		Censor-Hillel, Haeupler, Kelner, \& Maymounkov~\cite{CHKM12,CHKM12-journal} & $\tilde{O}(D)$ & $\tilde{O}(n)$ \\
		\hline
        Haeupler~\cite{H13,H15} & $\tilde{O}(D)$ & $\tilde{O}(n)$ \\
		\hline
        Ghaffari \& Kuhn~\cite{GK18} & $\tilde{O}(\sqrt{nD})$ & $\tilde{O}(1)$ \\
		\hline
        This work - Theorem~\ref{thm:rumorSpreadingGeneral} & $\tilde{O}(D + \sqrt n )$ & $\tilde{O}(1)$ \\
		\hline
	\end{tabular}
		}
	\label{table:results-diameter}
\end{table*}

\paragraph{Applications (see Section \ref{sec:apps})}
Our rumor spreading algorithms can also be used to directly compute certain aggregate functions efficiently in the $\GOSSIP-\CONGEST$ model. For example, it is easy to compute minimum (or maximum) by simply spreading and aggregating the minimum value (there is no congestion due to the aggregation of different values at nodes). 
A byproduct of our gossip algorithms is that they construct a \textit{spanning tree}
of the underlying network whose \textit{diameter} is bounded by the algorithm's run time. This can be used for various applications. For example, the spanning tree leads to gossip algorithms for other fundamental distributed computing problems such as leader election (where the goal is to elect a unique node in the graph as a leader) and a minimum spanning tree (MST) as well as computing aggregate functions such as \textit{sum} and \textit{count} \textit{exactly}.

\begin{corollary}
    There exist $\GOSSIP-\CONGEST$ algorithms solving spanning tree construction, leader election, and aggregate function computation  in $O(c \log n / \Phi_c)$ rounds with high probability, for every $c\geq1$ and where $\Phi_c$ is the weak conductance of the graph.
\end{corollary}

\begin{corollary}
    There exist $\GOSSIP-\CONGEST$ algorithms solving MST construction, leader election, and aggregate function computation, with high probability, in $\tilde{O}(D + \sqrt n )$ rounds.
\end{corollary}

\subsection{Technical Overview}
\label{subsec:intro-tech-challenges}

\paragraph{Rumor Spreading on Weak Conductance Graphs (see Section~\ref{sec:alg})} As noted earlier, one of the difficult challenges involved with spreading a message in a graph is dealing with communication bottlenecks. Consider again the example of a dumbbell graph. While it may be simple and fast to spread messages within each of the individual cliques, the difficulty arises in finding the very few edges that belong to the cut set between the cliques. This is where the powerful technique of graph sketching~\cite{AhnGM12a} turns out to be useful. By rapidly aggregating certain information within a clique, we can with high probability find one of the edges leading out of it.

More generally, instead of a clique, suppose we have a high conductance component, of conductance $\Phi_c$ for some known $c$. In fact, let us decompose the graph into disjoint sets of high conductance components with edges between them. In such a graph, we can rapidly spread information within each such high conductance component, then use sketching to find inter-component edges and continue to spread the rumor. One might assume then that if there are $c$ such components, it would only take $O(c T)$ time to spread the rumor over the whole graph, where $T$ is the total time to spread the rumor within a given high conductance component and subsequently perform graph sketching. However, an important subtlety of graph sketching stops us here.

That subtlety is that graph sketching requires centralized aggregation of information at a node. 
Thus, if we want to find the outgoing edge out of a \textit{super cluster} of the graph, consisting of two or more components/clusters with inter-cluster edges, we must aggregate the information from all nodes within these components, which may take time proportional to the diameter of that super cluster. Imagine the supergraph with these components as nodes and inter-component edges as edges. If such a supergraph forms a line, then it is easy to see that a naive implementation of rumor spreading may result in the super cluster containing the rumor only growing by one cluster with every instance of graph sketching. Thus, the algorithm may take $O((1 + 2 + 3 + \ldots + c-1) T) = O(c^2 T)$ time to spread the rumor to all nodes. Thus, we must be a bit more careful with the growth of the super cluster. 

To overcome the above issue, we control the diameter of the super clusters that perform graph sketching such that in phase $i$, only super clusters with diameter at most $2^i$ in the supergraph perform graph sketching. We complement this by showing that the size of the super clusters also grows exponentially, so that after $\log c$ phases, all clusters in the supergraph belong to one super cluster and thus the rumor has been spread to all nodes. The time then to perform this would be $O((1 + 2 + 2^2 + 2^3 + \ldots + 2^{\log c}) T) = O(c T)$.

\paragraph{Rumor Spreading on General Graphs
(see Section~\ref{sec:diameterAlg})} A natural approach for solving rumor spreading in $\GOSSIP-\CONGEST$ within general graphs (see, e.g., \cite{GK18}) is to sparsify the communication graph $G$. However, such a sparsification must run fast even with the lightweight communication of $\GOSSIP-\CONGEST$ --- that is, with communication over only $O(n)$ edges per round. 

A well-known sparsification technique consists of separating the nodes of $G$ into two sets according to some (parameter) degree threshold $\Deltath$, say $H$ for the high degree node subset and $L$ for the low degree node subset. (This can be done given only the nodes' initial knowledge of their degree.) Then, by definition the set of edges incident to $L$ is already somewhat sparse, since it contains $O(n \Deltath)$ edges. As for the remainder of $G$, a fast sparsification method is the following. First, sample uniformly at random some $1/\Deltath$ fraction of the nodes in $V$, which we call \emph{stars}, and have each high degree node join the cluster of any neighboring star node, if one exists. (Finding such a neighbor requires only sampling $O(\Deltath \log n)$ neighbors, hence is relatively fast even in $\GOSSIP-\CONGEST$.) These star clusters are then merged together using any classical tree-merging algorithm (e.g., Boruvka's algorithm), within which outgoing edges are found efficiently via graph sketches. In other words, we compute a maximal forest over the high degree parts of $G$, with one nice property: the sum of the diameter of the forest's trees is bounded by $O(n \log n/\Deltath)$ w.h.p.

The resulting sparse graph also, crucially, satisfies some \emph{stretch properties}: (1) large additive stretch of $O(n \log n/\Deltath )$ on shortest paths within the high degree parts of $G$, and (2) no stretch on the shortest paths within the low degree parts of $G$. However, the rumor does not spread efficiently throughout all parts of that sparse graph. Indeed, although the rumor spreads by one hop along the computed maximal forest, the rumor spreads more slowly along the edges incident to $L$ (with a multiplicative factor blow-up in in the degree $\Deltath$). As a result, the rumor spreading takes $\tilde{O}(\Deltath D + n / \Deltath)$ rounds. By optimizing $\Deltath$, one can trade off the slow down from a larger additive stretch with the slow down from the degree of nodes in $L$. Such an optimization leads, essentially, to the $\tilde{O}(\sqrt{nD})$ runtime obtained by \cite{GK18}.

In this work, we push this natural sparsification strategy one step further. More concretely, we sparsify and orient the subgraph in such a way that the maximum out-degree over the low degree parts is significantly reduced --- in fact, down to a maximum out-degree of $O(\log n)$. To do so, we accept one minor trade-off: a $O(\log n)$ multiplicative stretch on the (undirected) shortest paths over the low degree parts. Through an efficient simulation of the $\CONGEST$ model (which takes advantage of the low out-degree on the oriented subgraph \emph{and} the low stretch on its unoriented version), this additional sparsification step takes only $\tilde{O}(\Deltath)$ rounds, after which the additional properties satisfied by our sparsification technique directly ensure that rumor spreading takes only $\tilde{O}(D + n / \Deltath + \Deltath)$ rounds. Optimizing $\Deltath$ leads to our $\tilde{O}(D + \sqrt n)$ runtime.

\subsection{Related Work and Comparison}
\label{subsec:intro-related-work}

Gossip and rumor spreading protocols have been  studied
extensively for many decades starting with the work of Frieze and Grimmett~\cite{FG85}, and then Pittel~\cite{P87}. We restrict ourselves to those
that are most relevant to this work and refer to prior works~\cite{CHKM12,CS11,CS12,H13,H15} and the references therein for a good survey of related work.

Rumor spreading has been studied with respect to  spreading a rumor from a single node as well as spreading a rumor from each node in the network (or \textit{all-to-all} rumor spreading).
Much of the original research on rumor spreading considered
spreading a  \textit{single rumor} according  to the uniform gossip process on a \textit{complete} graph where it was shown that it took $O(\log n)$ rounds~\cite{FG85,P87}.  
The work of Karp,  Schindelhauer, Shenker, and Vocking~\cite{KSSV00} studied uniform and \textit{non-uniform} (push-pull) gossip algorithms for the complete graph and showed upper and lower bounds on the total message complexity with respect to the round complexity. 
Chierichetti, Lattanzi, and  Panconesi~\cite{CLP10-soda}  showed that \textit{uniform} gossip gives fast rumor spreading in expander graphs (which are graphs with constant conductance) and more generally in $O(\frac{\log^4 n}{\Phi^6})$ rounds in graphs of conductance $\Phi$. This was improved
to $\frac{\log^2 (\frac{1}{\Phi}) \log n}{\Phi}$ by the same set of authors~\cite{CLP10-stoc} and finally to  $O(\log n/\Phi)$ rounds by Giakkoupis \cite{Giakkoupis2011} (see also \cite{CGLP18}). This bound in terms of conductance is in general tight \cite{CGLP18,Giakkoupis2011}. 

It is important to note that the above uniform gossip algorithm uses small messages to spread a (single) rumor: in particular, only the rumor is exchanged over any edge. No IDs, etc., are sent as part of messages. Hence, according to our model (Section~\ref{subsec:intro-model}),
it works in the $\GOSSIP-\CONGEST$ model.

Since the performance of uniform gossip is tied to the conductance of the graph, its run time is slow when the conductance is small. Graph with bottlenecks (e.g., dumbbell) have small conductance, despite having a small diameter.  To address this, as mentioned earlier, Censor-Hillel and Shachnai~\cite{CS11,CS12} defined the notion of \textit{weak conductance} $\Phi_c$ and gave a gossip-based rumor spreading algorithm whose run time is determined by 
$\Phi_c$ (for all $c\geq 1$). For many graphs such as the dumbell, $\Phi_c$ is significantly smaller than the conductance $\Phi$. This algorithm is more sophisticated than uniform gossip and cleverly avoids sending too many messages to neighbors who have already received the rumor. But to accomplish this saving, the algorithm uses large-sized messages, up to linear in $n$, since a message can contain all the IDs of the nodes (informally, the IDs help a node to identify its neighbors who have already gotten the message without the need to communicate).  This is the situation when the goal is to  spread even a single rumor.
We improve on this work, both in terms of running time, and more importantly, by using small-sized messages.

Subsequent works~\cite{CHKM12,CHKM12-journal, H13,H15} presented gossip algorithms that ran in time that is independent of the conductance of the graph ---  in $O(D+\polylog{n})$ rounds (with high probability)~\cite{CHKM12,CHKM12-journal} and in $O(D\log n + \log^2 n)$~\cite{H13,H15}, where $D$ is the graph diameter. This is near optimal since $\Omega(D)$ is a trivial lower bound even for non-gossip algorithms (where a node can send a message to all of its neighbors in a round).

  Ghaffari and Kuhn~\cite{GK18} presented a gossip algorithm running in $\tilde{O}(\sqrt{nD})$ rounds to spread a rumor in the $\GOSSIP-\CONGEST$ model.  We improve over this work
 by presenting a gossip algorithm running in $\tilde{O}(D+\sqrt{n})$
 rounds in the $\GOSSIP-\CONGEST$ model.

\section{Preliminaries}
\label{sec:prelims}

We start by presenting the following standard Chernoff bounds, which we use in Section \ref{sec:diameterAlg}.

\begin{lemma}[Chernoff Bounds \cite{Upfalbook}]
\label{lem:ChernoffBound}
Let $X_1,\ldots,X_k$ be independent $\{0,1\}$ random variables. Let $X$ denote the sum of the random variables, $\mu$ the sum's expected value and $\mu_L, \mu_H$ be any value respectively smaller and greater than $\mu$. Then,
\begin{enumerate}
    \item $\Pr[X \leq (1-\delta) \mu_L] \leq \exp(- \delta^2 \mu_L / 2)$  for $0 \leq \delta \leq 1$,
    \item $\Pr[X \geq (1+\delta) \mu_H] \leq \exp(- \delta^2 \mu_H / 3)$ for $0 \leq \delta \leq 1$,
    \item $\Pr[X \geq (1+\delta) \mu_H] \leq \exp(- \delta^2 \mu_H / (2+\delta))$ for $\delta \geq 1$.
\end{enumerate}

\end{lemma}

Next, we give several definitions related to a relaxation of conductance. The relaxed conductance notion --- that of \emph{weak conductance} --- was first introduced in \cite{CS11,CS12}, and leads to a natural decomposition into clusters with weak conductance. That decomposition is implicitly used in \cite{CS11,CS12}, but we define it explicitly, as it both helps describe our algorithms and may be of independent interest.
Next, in Section \ref{subsec:graph-sketch}, we define graph sketches. Finally, in Section \ref{subsec:communication-primitives}, we describe several well-known distributed computing primitives, including some that use graph sketches, but adapted to the $\GOSSIP-\CONGEST$ model.

\subsection{Conductance and Weak Conductance}
\label{prelims:weak-conductance}

For any $S \subseteq V$, the conductance of the cut $(S,T=V\setminus S)$ is $\phi(S,T) = |cut(S,T)|/\min\{vol(S),vol(T)\}$, where $vol(S)$ is the sum of the degrees of all nodes in $S$, and $cut(S,T)$ denotes the set of all edges with one vertex in $S$ and one vertex in
    $T = V-S$. Then, the \emph{conductance} over the graph $G$ is naturally defined as:
$$ \Phi(G) = \min\limits_{S \subseteq V} \phi(S, V \setminus S) $$

The conductance $\Phi(G)$ captures how fast information can be spread from one node to all other nodes in graph $G$. If one only wants to spread the information to at least $n/c$ nodes (for some constant $c > 1)$, then a weaker definition can be used. More concretely, for any value $c \geq 1$, we can define the \emph{weak conductance} as in \cite{CS11,CS12}:
$$ \Phi_c(G) = \min\limits_{v \in V} \max\limits_{\substack{S_v \subseteq V, \\ S_v \ni v, \\ |S_v| \geq n/c}} \min\limits_{U \subseteq S_v} \phi(U, S_v \setminus U)$$

\begin{figure}
\includegraphics[page=2,width=0.8\textwidth]{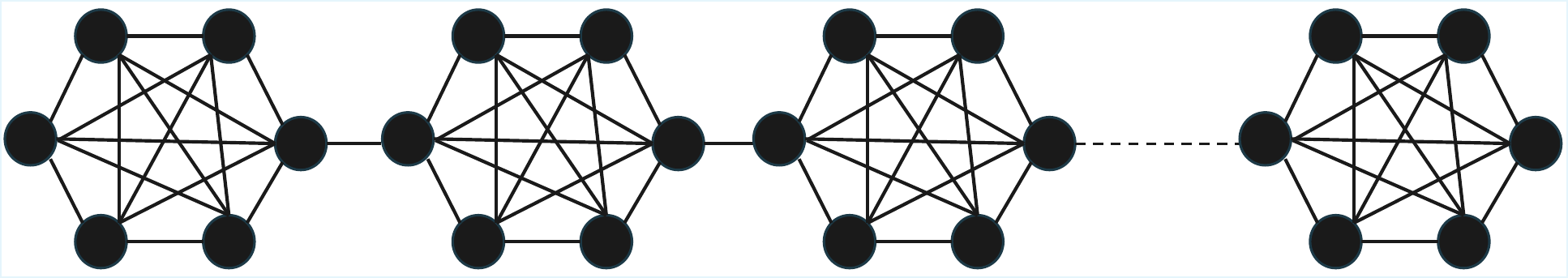}
\caption{A $c$-barbell of $n$ nodes. There are $c$ cliques, each containing $n/c$ nodes, connected in a path.}\label{fig:c-barbell} 
\Description[A $c$-barbell of $n$ nodes.]{There are $c$ cliques, each containing $n/c$ nodes, connected in a path.}
\end{figure}

Whenever $G$ is clear from the context, we simply say $\Phi$ and $\Phi_c$.
An illustrative example of the difference between conductance and weak conductance is to consider a $c$-barbell graph (see~\cite{CS11,CS12} and Figure~\ref{fig:c-barbell}) where $c$ cliques $C_1, C_2, \ldots, C_c$, each of size $n/c$, connected in a line 
such that between cliques $C_i$ and $C_{i+1}$, $1 \leq i \leq c-1$, there is one edge between a single node in $C_i$ and a single node in $C_{i+1}$. 
In this graph, $\Phi_c = \Theta(1)$ while the overall conductance is $\Phi = \Theta(c/n^2)$.

The above definition implies that any graph $G$ with weak conductance $\Phi_c$, for some value of the parameter $c \geq 1$, satisfies the following (structural) property. For each node $v$, there exists some maximal connected component that contains $v$, of size at least $n/c$, and with conductance at least $\Phi_c$. (Here, maximal is meant in the sense the component is not contained within any other such connected component.) We call the union of all such maximal components of $v$, the \emph{maximal spreading component} of $v$, and emphasize that via the result of Giakkoupis \cite{Giakkoupis2011}, broadcast (and information spreading for $\LOCAL$) can be done in $T= O(\log n / \Phi_c)$ rounds within that component.

For the algorithm in Section~\ref{sec:alg}, we assume that every node knows the values of both $c$ and $\Phi_c$. We note that the previous best algorithm for this setting~\cite{CS11,CS12} needs this assumption in order to achieve termination. 

\paragraph*{(Weak) Sunflowers and (Weak) Sunflower Decomposition.}
We provide a decomposition of the graph, with a terminology inspired from the sunflower set systems used in the famous sunflower lemma proved by Erdös and Rado in \cite{ER60}.

For any node $v$, we can consider all maximal spreading components (from any node $u \neq v$) that intersect with the maximal spreading component of $v$. We call the union of all such maximal spreading components the \emph{(weak) sunflower} of $v$. Moreover, we refer to the maximal spreading component of $v$ as its \emph{kernel}, and to the other intersecting components as its \emph{petals}. (Note that kernels and petals here differ from those in sunflower set systems, since different petals may intersect different subsets of the kernel.)
Weak sunflowers can be used to decompose the communication graph $G$, with major implications for the design of rumor-spreading algorithms in the $\GOSSIP-\CONGEST$ model. More precisely, we show that you can decompose the graph $G$ into at most $j \leq c$ sunflowers, associated respectively with nodes $v_1, v_2, ..., v_j$, such that the maximal spreading components of $v_1, v_2, ..., v_j$ are pairwise disjoint and the union of the sunflowers covers the whole graph. 

\begin{lemma}
\label{lem:sunflowerDecomposition}
    Any graph $G$ with weak conductance $\Phi_c$, for some value of the parameter $c \geq 1$, can be decomposed into at most $\lfloor c \rfloor$ pairwise kernel-disjoint sunflowers.
\end{lemma} 

\begin{proof}
   The decomposition can be obtained constructively as follows. Let $\mathcal{S}$ be a collection of sunflowers, initially empty. Then, take any arbitrary node $v_1$, compute its sunflower $S_1$ and add $S_1$ to $\mathcal{S}$. After which, as long as there remains some node outside $\mathcal{S}$, find such a node $v_i$ (i.e., some node $v_i \in V$ such that $v_i \not \in S, \forall S \in \mathcal{S}$), compute its sunflower $S_i$ in graph $G$ and add $S_i$ to $\mathcal{S}$. Note that this procedure computes pairwise kernel-disjoint sunflowers, otherwise the maximal spreading component of node $v_i$ would be in a prior sunflower. As such, for each sunflower in $\mathcal{S}$, each kernel must contain at least $\lceil n/c \rceil \geq n/c$ unique nodes, and thus $|\mathcal{S}| \leq c$. Since $|\mathcal{S}|$ is an integer, in fact it holds that $|\mathcal{S}| \leq \lfloor c \rfloor$. In other words, the procedure computes at most $\lfloor c \rfloor$ pairwise kernel-disjoint sunflowers that cover graph $G$.
\end{proof}

We emphasize that the decomposition is in general not unique. In particular, in the proof there may be many different choices for the node sequence $v_1, v_2, ..., v_j$, where $j \leq c$. Moreover, we also point out that although the sunflowers are (pairwise) kernel-disjoint, their petals (and thus the sunflowers themselves) may not be disjoint.

\subsection{Graph Sketches}\label{subsec:graph-sketch}

Consider an arbitrary set of nodes $V$, such that all nodes of $V$ have unique IDs in $[1,N]$, where $N$ is known to all nodes. 
We define, for any graph $G = (V,E)$ and node $v \in V$, the \emph{incidence vector} $\incVect(v) \in \mathbb{R}^{\binom{N}{2}}$ whose entries correspond to all possible choices of two IDs in $N$.  An entry in $\incVect(v)$  corresponding to the possible edge between $u$ with $id_u\in N$ and $v$ is $0$ if $(u,v) \notin E$, $1$ if $(u,v) \in E$ and $id_u > id_v$, and $-1$ otherwise (i.e., $(u,v) \in E$ and $id_u < id_v$). Entries in $\incVect(v)$ that correspond to a possible edge not including $v$ (e.g., $(u,w)$) have value $0$. 
Naturally, one can extend this definition to any node subset $S \subseteq V$: more precisely, $\incVect(S) = \sum_{v \in S} \incVect(v)$. Note that by linearity, the non-zero indices of $\incVect(S)$ indicate exactly which edges are in the cut of $S$ with respect to $G$, that is, in $E_G(S, V \setminus S)$. 

One may use the incidence vector of some node set $S$ to sample an outgoing edge from $S$, if one exists, uniformly at random from all such outgoing edges, i.e., sample a non-zero entry in $\incVect(S)$ uniformly at random. However, in a distributed setting, to compute the incidence vector on a set of nodes $S$, one would need to aggregate the incidence vectors of the nodes belonging to $S$.   
This is problematic since incidence vectors are exponentially larger (recall that they have size $\binom{N}{2}$) than our $O(\polylog N)$ message size. 

Fortunately, we can use a \emph{linear sketch} \cite{AhnGM12a} --- a well-chosen linear function from $\mathbb{R}^{\binom{N}{2}}$ to $\mathbb{R}^k$ --- or in other words, a \emph{graph sketching matrix}  --- a well-chosen $k \times \binom{N}{2}$ size matrix $M_G$ --- to compress these vectors $\left( \incVect(v) \right)_{v \in V}$ into smaller (sketch) vectors $\left( \sketchVect(v) \right)_{v \in V}$ of size $k = O(\polylog{\binom{N}{2}})$; more concretely, $M_G \cdot \incVect(v) = \sketchVect(v)$. 
Moreover, although this compression necessarily loses some information, it has two major advantages. First, it is possible to sample (almost) uniformly at random one of the non-zero indices of $\incVect(v) \in \mathbb{R}^{\binom{N}{2}}$ by performing operations on $\sketchVect(v)$ only (albeit with some small failure probability). More concretely, for any graph sketching matrix $M_G$ and for any subset $S \subseteq V$, there exists a sampling function $f_G$ that takes the sketch vector $\sketchVect(S)$ as input and outputs an edge chosen uniformly at random in $E_G(S,V \setminus S)$. 
Second, the linearity of the graph sketching matrix allows us to compute $\sketchVect(S)$ without computing $\incVect(S)$, but instead by 
computing the sketch vectors $\left( \sketchVect(v) \right)_{v \in S}$ and summing them. 

In summary, for any subset $S \subseteq V$, we can sample an edge chosen uniformly at random in $E_G(S,V \setminus S)$ by (i) having nodes agree on $\Theta(\log^2 n)$ true random bits that they can use to locally compute a common graph sketching matrix $M_G$ with polynomially bounded integer entries~\cite{JST11, PRS18}, (ii) aggregating the sketch vectors of all nodes in $S$, and (iii) applying the sampling function $f_G$ on the aggregate vector $\sketchVect(S)$. Note that these steps require messages of small $O(\polylog n)$ size only. 
A more formal statement is given below, which can be obtained straightforwardly by adapting known results~\cite{AhnGM12a,JST11,PRS18}.

\begin{lemma}\label{lem:graph-sketch-works}
    For any upper bound $N$ on the ID range and constant $0 < \delta < 1$, there exist a graph sketching matrix $M_G$ (with entries polynomially bounded in $N$) and a sampling function $f_G$ such that for any node subset $S \subset V$, the aggregate sketch vector $\sketchVect(S) = \sum_{u \in S} \sketchVect(u)$ can be represented using $O(\polylog n)$ bits, and $f_G(\sketchVect(S))$ samples a (uniformly) random edge in $E_G(S,V \setminus S)$ with probability $1 - \delta$.
\end{lemma}

\subsection{Communication Primitives in the $\GOSSIP-\CONGEST$ Setting}
\label{subsec:communication-primitives}

The following ``basic'' distributed primitives (e.g., tree construction, broadcast and convergecast over a tree, etc.) are well-known in other models, such as $\CONGEST$. However, we describe them again here because their description needs to be adapted for the $\GOSSIP-\CONGEST$ model. 

We remind that a (distributed) tree is a set of undirected edges of $G$ that forms a tree, and such that each node in the tree knows which of its incident edges is in the tree as well as which one of these leads to its parent in the tree. (As a result, each node also knows which edge leads to its children in the tree.) Similarly, a (distributed) forest is a union of disjoint (distributed) trees.

Moreover, we assume that rumors follow some agreed upon, total ordering. For example, when rumors are nodes' IDs, the natural total order on IDs is such an ordering. This allows nodes to perform certain aggregation operations on rumors, and thus avoid being slowed down by the limited bandwidth of the $\GOSSIP-\CONGEST$ model. 

\paragraph*{Tree Construction.}

In graphs with conductance $\Phi$, rumor-spreading can be done via (a simple modification of) uniform gossip. Recall that in uniform gossip, each node contacts a random neighbor every round, and exchanges the rumor (if any of them have it). Then, the modification is simple. We allow a set of sources nodes, and thus multiple rumors, and we modify the algorithm so that each node exchanges the highest rumor it has seen until now upon every contact (initiated or received). In which case, $O(\log n /\Phi)$ rounds suffice for the highest rumor to spread (with high probability) over the whole graph, in the $\GOSSIP-\CONGEST$ model. 

\begin{lemma}[\cite{Giakkoupis2011}]
\label{lem:Giakkoupis}
    For any $n$-node graph $G=(V,E)$ with conductance $\Phi$, and any set $S \subseteq V$ of source nodes. Then, the above algorithm spreads the highest rumor to all other nodes in $O(\log n /\Phi)$ rounds, with high probability.
\end{lemma}

Moreover, another simple modification gives a spanning tree construction algorithm in the $\GOSSIP-\CONGEST$ model. Assume each source node holds a unique rumor (e.g., an ID). Then, it suffices if whenever a node receives a rumor higher than any it has seen until now, the node takes the sender as its parent (in the tree being constructed). If the highest rumor the node has seen when the rumor-spreading algorithm terminates is the rumor it initially held (if it held any), then that node becomes the root. This simple, modified version of uniform gossip constructs a (distributed) spanning tree of $G$ (as defined previously) rooted in the source node with the highest rumor.

\begin{corollary} 
\label{lem:treeConstruction}
    For any $n$-node graph $G=(V,E)$ with conductance $\Phi$, and at least one source node, the above algorithm constructs a spanning tree rooted at the source node with the highest rumor in $O(\log n /\Phi)$ rounds, with high probability.
\end{corollary}

\begin{proof}
    By Lemma \ref{lem:Giakkoupis}, each node $v \in V, v \neq s$ receives the highest rumor (and thus some rumor strictly higher than its own, if it initially held any), in $O(\log n /\Phi)$ rounds, with high probability. (For the remainder of the proof, we condition on this event happening.) Hence, by the algorithm's description, each node except the root orients one of its incident edges (toward its parent) upon receiving that highest rumor. Since each node orients the corresponding edge towards another node having received the (same, highest) rumor strictly earlier than itself, these oriented edges cannot form a cycle. The statement (and error probability) follows.
\end{proof}

If we consider graphs in which only smaller regions have good conductance, such as weak conductance graphs, then the above modified (highest rumor) uniform gossip algorithm allows fast rumor spreading to happen within sunflowers --- this essentially follows by Lemma \ref{lem:Giakkoupis}.

\begin{corollary}
\label{cor:sunflowerSpreading}
    Consider some constant $c \geq 1$, $n$-node graph $G=(V,E)$ with weak conductance $\Phi_c$, and source node $s \in V$. 
    Then, after $O(\log n /\Phi_c)$ rounds, it holds with high probability that any node $v$ in the sunflower of $s$ receives some rumor greater or equal to the highest rumor initially held by a node in that sunflower.
\end{corollary}

\begin{proof}
    It is straightforward to see that Lemma \ref{lem:Giakkoupis} also implies that for any subgraph $H = (V_H,E_H)$ of $G$ with conductance $\Phi'$ (with at least one source node), the (modified) uniform gossip algorithm ensures that each node $v \in V_H$ receives the highest rumor initially held by some node in $V_H$, or some even higher rumor from $V \setminus V_H$, in $O(\log n /\Phi')$ rounds and with high probability.

    Now, for any node $u$ in the maximal spreading component of $s$, there exists at least one connected component $C$ containing $u$ and $s$, of size at least $n/c$ and with conductance at least $\Phi_c$. Then, by the above statement, any node in $C$, and thus $u$, receives the highest rumor initially held by some node in $C$,
    or some even higher rumor, in $O(\log n /\Phi_c)$ rounds and with high probability. We can apply the same reasoning, but now we consider any node $v$ in the sunflower of $s$. Then, some maximal spreading component intersecting that of $s$ must contain $v$, and thus there exists some node $u$ in that intersection, and at least one connected component $C'$ containing $u$ and $v$, of size at least $n/c$ and with conductance at least $\Phi_c$. Then, again by the above statement, any node in $C'$, and thus both $u$ and $v$, receives the highest rumor initially held by some node in $C'$, or some even higher rumor, in $O(\log n /\Phi_c)$ rounds and with high probability.

    Finally, let $w$ be the node holding initially the highest rumor in the sunflower of $s$. Then, $s$ receives the rumor of $w$, or some even higher rumor from outside the sunflower, in $O(\log n /\Phi_c)$ rounds and with high probability, by the above. In turn, this implies that after another $O(\log n /\Phi_c)$ rounds and with high probability, any node in the sunflower of $s$ receives the rumor of $w$, or an even higher rumor.
\end{proof}

However, in weak conductance graphs, a single rumor does not necessarily reach all nodes within $O(\log n /\Phi_c)$ rounds. In fact, even within the sunflower of some node $s$, although the above corollary implies that all nodes receive some rumor greater or equal to that of $s$, it may still happen that different nodes receive different rumors. As such, the above (tree construction) modification of uniform gossip does not give a direct analogue of Corollary \ref{lem:treeConstruction} for weak conductance graphs, or in other words, uniform gossip on weak conductance graphs does not lead to a simple and time-efficient spanning tree construction algorithm.

Yet, for any weak conductance graph $G$ and any set $S \subseteq V$ of source nodes, we can time-efficiently (i.e., in $O(\log n /\Phi_c)$ rounds) construct a forest rooted in $S$ and spanning the union of the sunflowers of nodes in $S$, in the $\GOSSIP-\CONGEST$ model. This comes with a major caveat: the obtained forest may have depth $\omega(\log n / \Phi_c)$, despite the runtime being bounded to $O(\log n /\Phi_c)$. This is because even though a single rumor can create a branch of at most $O(\log n /\Phi_c)$ nodes in the computed forest, one branch may result from the propagation of multiple, different rumors.

\begin{corollary}
\label{cor:forestComputation}
    Consider some constant $c \geq 1$, $n$-node graph $G=(V,E)$ with weak conductance $\Phi_c$, and set $S \subseteq V$ of source nodes. Let $U$ be all nodes contained in the union of the sunflowers of nodes in $S$.
    Then, $O(\log n /\Phi_c)$ rounds suffice to compute some (distributed) forest spanning $U$, with high probability.
\end{corollary}

\begin{proof}
    By Corollary \ref{cor:sunflowerSpreading}, it holds with high probability that each node $v \in U$, $v \notin S$, receives some rumor in $O(\log n /\Phi_c)$ rounds. (For the remainder of the proof, we condition on this event happening.) Moreover, nodes in $S$ may receive some rumor strictly higher than their own. Any node in $S$ for which this does not happen, becomes a root by the algorithm description.
    
    Moreover, once again by the algorithm's description, each non-root node orients one of its incident edges (toward its parent) for the last time upon receiving the highest rumor it sees in the $O(\log n /\Phi_c)$ rounds. Now, it follows that each node orients that edge towards some node having received the same rumor in some strictly earlier round, or a strictly higher rumor. In turn, this implies that these oriented edges cannot form a cycle. The statement (and error probability) follows.
\end{proof}

\paragraph*{Broadcast \& Convergecast.}

For these two primitives, consider some distributed forest of root set $S \subseteq V$ and of depth $T$, where $T$ is some known upper bound. 

Then, we describe a broadcast primitive over that tree. We assume that any root $s \in S$ has some message $\mu_s$, and that broadcast is successful if for any root $s \in S$, every node in the tree of $s$ outputs $\mu_s$. This can be done in $T$ rounds, during each of which, each node contacts its parent and receives the broadcast message if its parent already knows it. 

A simple modification yields a primitive for a partial broadcast up to depth $1 \leq d  \leq T$, described concretely as follows: for any root $s \in S$, every node in the tree of $s$ and at depth at most $d$ must output $\mu_s$. Indeed, running the broadcast primitive for $d$ rounds ensures that every node at depth at most $d$ receives the corresponding broadcast message within these rounds.
 
The convergecast primitive works conversely to the broadcast primitive. Now, each node $v$ in the forest may start with some message $\mu_v$, and all nodes know some aggregation function (e.g., min, max, sum and average) for these messages. Then, in each of the $T$ rounds, once a node $v$ has received a message from each of its children, it aggregates these messages with its own message $\mu_v$. If the node is not the root, then it contacts its parent, and sends it that aggregate. Otherwise, $v$ outputs the aggregate.

\begin{proposition}  
\label{prop:treePrimitives}
    For any $n$-node graph $G$ and depth-$T$ (distributed) forest, where $T$ is some known upper bound, broadcast and convergecast over the forest can be done in $O(T)$ rounds. Moreover, for any integer $1 \leq d \leq T$, partial broadcast up to depth $d$ can be done in $O(d)$ rounds.
\end{proposition}

\paragraph*{Sampling Outgoing Edges.}

Once again, consider some distributed forest with root set $S \subseteq V$ and of depth $T$, where $T$ is some known upper bound. We describe how the tree roots can, simultaneously, each sample an outgoing edge (i.e., whose other endpoint is not in the tree), if there exists one, in $O(T)$ rounds and with high probability. This is done by relying on graph sketches, which are described in Section \ref{subsec:graph-sketch}.

A tree root creates a common random string of $O(\polylog n)$ bits and broadcasts that string to all nodes in its tree (using the above described broadcast primitive). After which, each node in the tree creates some $O(\log n)$ graph sketches (using their incidence vectors and that common random string) and convergecasts them to the root. (Note that during the convergecast, graph sketches are aggregated, but since each node has a single parent in the tree, the same graph sketch is not aggregated twice.) Finally, the root samples some outgoing edge (with high probability) and broadcasts that edge to all nodes in its tree.

\begin{proposition}  
\label{prop:gossip-sampling-edges}
    For any $n$-node graph $G$ and depth-$T$ forest, where $T$ is some known upper bound, each (distributed) tree's roots can sample an outgoing edge in $O(T)$ rounds, with high probability. 
\end{proposition}

\section{Rumor Spreading in Graphs with Large Weak Conductance}
\label{sec:alg}

Any graph with weak conductance $\Phi_c$, for some value of the parameter $c \geq 1$, can be decomposed into at most $\lfloor c \rfloor$ pairwise kernel-disjoint sunflowers (see Lemma \ref{lem:sunflowerDecomposition}). Combining this structural property with graph sketching techniques, we solve rumor spreading (as well as spanning tree construction and leader election) while using only $O(\polylog n)$-sized messages. Moreover, we do so with a time complexity of $O(c \log n / \Phi_c)$ rounds, which improves on that of \cite{CS12} (where the latter requires messages with size linear in $n$).

Note that as a result, we also show that for graphs with good weak conductance for small enough $c \geq 1$, that is, such that $c/ \phi_c = O(\polylog n)$, we can obtain fast (i.e., $O(\polylog n)$ rounds) and message-efficient (i.e., $\tilde{O}(n)$ messages) spanning tree construction (and leader election) even in the $\GOSSIP-\CONGEST$ model. 

\subsection{Super Cluster Primitives}
\label{subsec:superset-primitives}

Before presenting our rumor spreading algorithm for weak conductance graphs, we first give several more basic communication primitives. These primitives cope with the fact that we cannot afford to merge several clusters (and their spanning trees) into a single one for much of the algorithm presented in Section \ref{subsec:weakConductanceAlg}, otherwise its runtime would grow significantly. Hence, we define and leverage \emph{super clusters} --- these are collection of clusters with good communication properties in $\GOSSIP-\CONGEST$. 

More formally, a super cluster of \emph{size} $k \geq 1$ and \emph{depth} $T$ is a triple containing (1) a set of some $k \geq 1$ clusters, (2) a set of $k$ disjoint trees, each spanning one of the clusters and of diameter at most $T$, and (3) a set of (undirected) inter-cluster edges. Moreover, the inter-cluster edges must satisfy the following properties. First, for each inter-cluster edge, at least one of the endpoints is responsible for communication (i.e., only this node is allowed to contact the other endpoint), and that endpoint node knows that it is responsible for that edge. Second, no two inter-cluster edges are incident, i.e., each node is responsible for the communication of a single incident edge. 

Next, we implement several basic distributed primitives for super clusters time-efficiently. To do so, we show how any super cluster can efficiently simulate a communication-restricted distributed algorithm on the following \emph{cluster supergraph}: each cluster is a (super) node, and two clusters (or super nodes) are connected if some inter-cluster edge has its endpoints in both clusters. More concretely, we provide the simulation for synchronous distributed algorithms. Such algorithms proceed in iterations (to distinguish from the rounds used by the nodes in the original communication graph), and within each iteration, a super node (1) executes some local computation, where ``local'' here means $O(T)$-round operations restricted to the cluster (e.g., broadcast, convergecast, etc.), (2) sends the same message to all its (super node) neighbors (i.e., broadcast communication), and (3) receives an aggregate (e.g., bitwise OR, average, sum, min, max, etc.) of the messages sent by its neighbors.
Then, a single iteration on the supergraph can be simulated via $O(T)$ rounds in the original communication graph, by leveraging the super cluster's structure.
That is, to simulate one iteration, first each cluster's root computes what message it would send on the supergraph, and then the root broadcasts that message over its cluster, in $O(T)$ rounds. Second, all inter-cluster edges are used to communicate; for any inter-cluster edge $e$ incident to any node $v$ in the cluster, if $v$ is responsible for communication over $e$, then $v$ contacts the other endpoint of $e$ after receiving the broadcast from its cluster root.
(This takes a single round, since each node is responsible for at most one incident edge.) Third and finally, each cluster convergecasts, and in doing so, aggregates the information received by any of its nodes from other clusters.  

\begin{proposition}  
\label{prop:superclusterSimulation}
    Consider any $n$-node graph $G$, any super cluster of size $k$ and depth $T$ (where $T$ is some known upper bound) and any $T_I$-iteration algorithm $\mathcal{A}$ designed in the above described model. Then, the super cluster simulates the execution of $\mathcal{A}$ on the cluster supergraph, in $O(T \cdot T_I)$ rounds. 
\end{proposition}

With the above simulation, we can implement some basic primitives for Section \ref{sec:alg}. The first \emph{flooding} primitive is quite straightforward. All clusters have their own inputs, either some IDs or some boolean, as well as an integer input $\beta \geq 1$. Then, the primitive simulates a basic $\beta$ round flooding algorithm on the cluster supergraph. If the initial inputs are IDs (respectively, booleans), then in that supergraph, each node simply sends the highest ID (respectively, the OR of all booleans) seen until now to all of its neighbors, and upon aggregating any received IDs (respectively, booleans), keeps the maximum one (respectively, the OR result). Since the algorithm on the cluster supergraph trivially takes $O(\beta)$ iterations, the primitive on $G$ takes $O(\beta T)$ rounds (by Proposition \ref{prop:superclusterSimulation}).  

\begin{proposition}  
    Consider any $n$-node graph $G$, and any super cluster of size $k$ and depth $T$ (where $T$ is some known upper bound). For any integer $\beta \geq 1$, each cluster (root) can receive an aggregate of the values of all clusters within distance $\beta$ on the cluster supergraph, in $O(\beta T)$ rounds. 
\end{proposition}

This flooding primitive is a basic building block for the remaining primitives on super clusters. For example, given some integer input $\beta \geq 1$, say we want to detect if the eccentricity of the cluster with maximum root ID, within the cluster supergraph, is at most $\beta$. (We use this problem as a simple and efficient alternative to the problem of detecting whether the diameter of the cluster supergraph is at most $\beta$.) This distributed detection problem can be defined as follows (on the cluster supergraph). If that eccentricity is most $\beta$, then the cluster with maximum root ID should output true whereas all other clusters should output false. Otherwise, all clusters should output false.
This distributed detection problem can be solved easily, even within the restricted distributed model we consider for algorithms running on the cluster supergraph. More precisely, we first run the flooding primitive, using the cluster root IDs, for $\beta+1$ iterations. If any cluster sees its maximum highest ID seen until now change between the $\beta$th iteration and the last, $(\beta+1)$th iteration, then that cluster detects the eccentricity in question is strictly higher than $\beta$. Next, we run the flooding primitive again, but for $\beta$ iterations and where the initial input of each cluster is a boolean --- set to true if the cluster detected the above, and false otherwise. Finally, the output is decided as follows: any cluster (root) that did not see any higher ID during the first flooding primitive, nor received any true boolean value during the second flooding primitive, sets its output value to true.

\begin{proposition}  
    Consider any $n$-node graph $G$, and any super cluster of size $k$ and depth $T$ (where $T$ is some known upper bound). For any integer $\beta \geq 1$, the cluster with maximum root ID cluster can detect if its eccentricity in the cluster supergraph is strictly higher than $\beta$ or not, in $O(\beta T)$ rounds. 
\end{proposition}

Now, consider that the previous primitive returns that the eccentricity of the cluster with maximum root ID, within the cluster supergraph, is at most $\beta$, for some integer $\beta \geq 1$. Then, we can time-efficiently merge all clusters into a single one in the $\GOSSIP-\CONGEST$ model. More precisely, for all nodes in the super cluster, we modify its parent, if needed, such that the corresponding edges form a single spanning tree of diameter at most $\beta T$ covering all of the super cluster's nodes.
To do this, it suffices to run the flooding primitive (with IDs) twice, where nodes learn during the first flooding the maximum root ID in the super cluster, and learn during the second flooding which node is their parent in the new tree (i.e., the node from which they receive the ID corresponding to the maximum root ID from the first flooding).

\begin{proposition}  
\label{lem:reorientSuperCluster}
    Consider some $n$-node graph $G$ and super cluster of size $k$ and depth $T$ (where $T$ is some known upper bound), and let $U$ denote the nodes in the super cluster. If the eccentricity of the cluster with maximum root ID is at most $\beta$, for some integer $\beta \geq 1$, in the cluster supergraph, then we can construct a spanning tree of $U$ with diameter at most $2 \beta T$, in $O(\beta T)$ rounds.
\end{proposition}

Given such a spanning tree, we can time-efficiently (i.e., in $O(\beta T)$ rounds) compute an outgoing edge in $\GOSSIP-\CONGEST$ --- by directly using primitives from Section \ref{subsec:communication-primitives}. (Where here, outgoing means an edge leading to some node not in the super cluster, if any exists.) Moreover, note that for any such super cluster (i.e., fulfilling the above eccentricity condition), we can compute a temporary spanning tree. Indeed, if we make sure that nodes keep in memory the previous edges, then as a result, after some fixed number of rounds known to all nodes (e.g., once we've sampled an outgoing edge of the super cluster), we can revert back to the original super cluster (i.e., with spanning tree for each cluster, and inter-cluster edges).

\subsection{Algorithm}
\label{subsec:weakConductanceAlg}

Consider some communication graph $G$ with weak conductance $\Phi_c$, for some value of the parameter $c \geq 1$, and let $T(n) = O(\log n / \Phi_c)$ be an upper bound (whose only dependency on $c$ is captured by $\Phi_c$) on the runtime of the primitives in Section \ref{subsec:graph-sketch}
for a graph of $n$ nodes and weak conductance $\Phi_c$. 
Also, let $c^* = \lfloor c \rfloor$.  
We give a spanning tree construction algorithm that takes $O(c \log n / \Phi_c)$ rounds. 

We separate the algorithm into two parts: a set-up part and a tree-merging part. In the set-up part, we ``compute'' a decomposition of $G$ into at most $c^*$ pairwise kernel-disjoint sunflowers, but in a distributed fashion. More precisely, we compute a decomposition of $G$ into at most $c^*$ disjoint clusters --- each corresponding to one such sunflower --- such that each cluster is spanned by a tree of $O(\log n / \Phi_c)$ depth. A subtle point is that we do not compute the decomposition into sunflowers exactly; indeed, this would give us non-disjoint clusters whereas we need disjoint clusters for the next, tree-merging part. Instead, we carefully assign nodes in the intersection of the sunflowers (which are pairwise kernel-disjoint, but not pairwise disjoint) to a single cluster, without increasing the cluster's diameter.

After which, we take these clusters to be the initial clusters of the tree-merging part, which will combine the clusters' spanning tree into a single tree spanning the whole graph $G$. To merge these clusters fast, we leverage graph sketching techniques. Moreover, we merge clusters in a slightly unusual way. Both techniques combined together lead to our $O(c \log n / \Phi_c)$ runtime.\footnote{Running most existing merging techniques 
for few, say $o(\log n)$ phases, incurs a $\Omega(\log^* n)$ multiplicative overhead, if not $\Omega(\log n)$. Using such a technique here would result in $\omega(c \log n / \Phi_c)$ runtime.}

\paragraph*{Set-up.}

In the set-up part, we run $c^*$ phases. Initially, all nodes are uncovered (i.e., $C_0 = \emptyset$), and at the end of each phase $i$, the covered nodes (denoted by $C_i$) form a (distributed) forest $F_i$ of at most $c^*$ (disjoint) trees, and of depth at most $2T(n)$. (We remind that this means each node in a tree knows which edge leads to its parent, and additionally, we ensure each node knows its depth in the tree and the root's ID.) Let $U_i = V \setminus C_i$ denote the uncovered nodes at the end of phase $i$.
Each phase takes $6T(n)$ rounds, and can be decomposed into the following three steps. 

In the first step, all nodes run (the modified, tree construction) uniform gossip for $2 T(n)$ rounds and construct a (distributed) forest $F'$ spanning the union of the sunflowers of nodes in $U_{i-1}$ (see Corollary \ref{cor:forestComputation} in Section \ref{subsec:communication-primitives}) --- unless $U_{i-1} = \emptyset$, in which case no forest ends up being constructed.
However, note that this forest may have depth $\omega(T)$.

In the second step, we prune the forest $F'$. First, nodes run (modified, highest rumor) uniform spreading for $2T(n)$ rounds, where each node's initial rumor is the highest ID received during the first step. Any root of $F'$ that receives some higher ID becomes inactive for the remainder of this step; doing this will remove, from $F'$, any tree that doesn't span its root's sunflower (whp). After which, for all active roots in $F'$, we remove from their tree all nodes that are at depth strictly more than $2T(n)$. To do so, all active roots in $F'$ (partial) broadcast simultaneously (see Section~\ref{subsec:communication-primitives} for a description of the primitive) their ID over their tree, during the phase's next $2T(n)$ rounds. On the one hand, any node that receives any message over these $2T(n)$ rounds now knows the root's ID and can compute its depth within the tree (from the round in which it receives the message). In which case, we say it \emph{joins} that tree. On the other hand, any node that did not receive a message (during these last $2T(n)$ rounds) does not keep the edge to its parent computed in this phase's first step, and as such does not join any tree in this phase. (Note that it may be part of some tree it joined in some previous phase.) Whenever a node joins a tree for the first time, it becomes covered. To summarize, the second step prunes forest $F'$ into some forest $F''$ (rooted in the active nodes only), spanning less nodes, but with depth $2T(n)$ and such that each tree spans its root's sunflower.

In the third and final step, we merge the forest $F_{i-1}$ obtained at the end of the previous phase, together with the forest $F''$ computed in the current phase, by pruning away some unnecessary edges. More precisely, each node $v$ with multiple parents keeps the one with smallest depth (or more precisely, the one corresponding to the tree in which $v$ has smallest depth), and after that, the phase ends.
Note that this third and last step ensures that the resulting forest $F_i$ has depth $O(\log n / \Phi_c)$, whereas a naive approach would lead to depth $O(c \log n / \Phi_c)$.

\paragraph*{Tree-Merging.}

After the set-up part, we have at most $c^*$ (disjoint) clusters such that each cluster is spanned by a tree of depth $2T(n) = O(\log n / \Phi_c)$. Progressively, we form larger sets of connected clusters, which we call \emph{super clusters}. (Due to efficiency reasons, we maintain a loose structure for each super cluster --- that is, we do not form a single tree spanning all nodes in the super cluster --- until all phases of the merging are done. In what follows, we use primitives on super clusters described in Section~\ref{subsec:superset-primitives}.) More precisely, we run $\log c^*$ phases, and each phase $i$ ensures that the (disjoint) super clusters each contain at least $2^i$ clusters. After these $\log c^*$ phases, a single super cluster contains all clusters, and it is straightforward to merge these together into a single spanning tree (see the primitive in Section~\ref{subsec:communication-primitives}) as well as use that tree to perform rumor spreading (via the broadcast and convergecast primitives from Section~\ref{subsec:communication-primitives}).  

Next, we describe how the super clusters merge in every phase, and how this is done in a way such that over all phases, only $O(c \log n/\Phi_c)$ rounds suffice.
Initially, each cluster is its own super cluster. Then, in phase $i \in \{1,\ldots,\log c^*\}$, each super cluster (roughly) checks if it contains at least $2^i$ clusters, in which case it won't merge in this phase.
More precisely, if we define the cluster supergraph as the (virtual) graph in which each cluster is a (super) node, and two clusters (or super nodes) are connected if some inter-cluster edge has its endpoints in both clusters, then the super cluster performs the following simpler and more efficient procedure: it checks whether all clusters are within $2^i$ distance, within the cluster supergraph, of the cluster with maximum root ID. If yes, then the super cluster samples an outgoing edge, and does so in $O(2^i T(n))$ rounds (where the constants hidden in $O$ do not depend on $i$, and for more details, see Propositions~\ref{lem:reorientSuperCluster} and~\ref{prop:gossip-sampling-edges}).
Otherwise, the super cluster does nothing except possibly receive messages from other super clusters' outgoing edges. (In this case, the super cluster simply waits for a phase that matches its number of clusters. While waiting, the super cluster may grow via the addition of some new clusters, which happens when clusters outside the super cluster find an outgoing edge leading to the super cluster.)

\subsection{Analysis}
We begin with several lemmas regarding the set-up part, and in particular we show Lemma \ref{lem:setUpGuarantee}. That is, after the set-up part, with high probability, we have decomposed $G$ into ``few'' disjoint clusters, and each has ``small'' diameter.

\begin{lemma}
\label{lem:setupInvariant}
    Consider any phase $i \in \{0,\ldots,c^*\}$ of the set-up part. Then, forest $F_i$ spans all covered nodes and has depth $2 T(n)$. Moreover, with high probability, it holds that: 
    \begin{enumerate}
        \item For each root in $F_i$, all nodes within its sunflower are covered. 
        \item For any two roots in $F_i$, their sunflowers are pairwise kernel-disjoint. 
    \end{enumerate} 
\end{lemma}

\begin{proof}
    We prove the statement by induction on $i$. For the base case, note that no node is covered initially, and the forest is initially empty, hence the claim is trivially true.
    Next, we assume the induction hypothesis is true for some $i \in \{0,\ldots,c^*-1\}$ and prove the induction step. (We prove that the induction step holds with high probability, and this coupled with a union bound over the $c^* \leq n$ phases implies the lemma statement.) 

    To do so, first note that the first two steps compute a forest $F'$ of depth $2 T(n)$ over the nodes that become covered in phase $i+1$. By Lemma \ref{cor:forestComputation}, the first step constructs a forest. Then, the second step constructs $F''$ by removing some edge from $F'$, thus $F''$ is also a forest, and ensures that only nodes up to depth $2T(n)$ in $F'$ are kept for $F''$. Indeed, only these nodes receive the partial broadcast message by Proposition \ref{prop:treePrimitives}. Finally, any node that becomes covered in phase $i+1$ is, by definition, spanned by $F''$. Next, note that the induction hypothesis for phase $i$ implies that any node covered in phase $i$ or earlier is part of $F_i$, and that $F_i$ has depth at most $2T(n)$. 
    
    Now, recall that the algorithm description forces each covered node $v$ whose parent in $F_i$ differs from that in $F''$, to choose the one with smallest depth for $F_{i+1}$ (which $v$ deduces from its own depth in the corresponding forest). Let us assume by contradiction that we have a cycle in $F_{i+1}$. However, by the above procedure, each node's depth must be higher than that of its (possibly updated) parent. Since this cannot hold in a cycle, we get a contradiction and thus $F_{i+1}$ is a forest spanning the covered nodes at the end of phase $i+1$. Moreover, its depth is $2 T(n)$ because the depth of a node can only remain the same, or decrease, from $F_i \cup F''$ to $F_{i+1}$. 

    We now show that the two items hold with high probability for $F_{i+1}$. By the induction hypothesis for $i$, the two items respectively hold for any root and any two roots in $F_i$. First, let us assume that the first item does not hold for any root in $F''$, say root $r$. Then, it holds that some node $v$ in the sunflower of $r$ receives some higher ID (than that of $r$) from some other root node $r' $ of $F'$ during the first step (but $v$ does not join any tree in the second step). However, by Corollary \ref{cor:sunflowerSpreading}, root $r$ receives this higher ID, or an even greater one, during the first $2T(n)$ rounds of the second step with high probability. Thus, $r$ becomes inactive, which contradicts the fact that $r$ is a root of $F''$. As a result, we get that for any root in $F''$, and in fact, for any root in $F_{i+1}$, all nodes within its sunflower are covered.  

    Next, take any root $r$ in $F''$, and another root $r'$ from either $F_{i}$ or $F''$. Assume their sunflowers are not pairwise kernel-disjoint. If $r'$ is a root of $F_i$, then by item 1 for $F_{i}$, $r$ must be covered by the end of phase $i$, which leads to a contradiction with $r$ being a root of $F''$. Otherwise, if $r'$ is a root of $F''$, then by Corollary \ref{cor:sunflowerSpreading}, the node with the smallest rumor among $r$ and $r'$ receives a rumor greater than its own during the first step, and thus does not become a root of $F'$ (and thus of $F''$ also), which is also a contradiction. Hence, we get that the sunflowers of $r$ and $r'$ are pairwise kernel-disjoint with high probability, and thus the induction step holds with high probability. As mentioned above, this implies the lemma statement holds with high probability also.
\end{proof}

\begin{proposition}
\label{prop:coveredNodes}
    At the end of any phase $i \in \{0,\ldots,c^*\}$ of the set-up part, the number of uncovered nodes is at most $n - i \cdot n/c$, with high probability. 
\end{proposition}

\begin{proof}
    By Lemma \ref{lem:setupInvariant}, the pairwise kernel-disjoint sunflowers of the forest's roots (in phase $i$) contain only covered nodes, with high probability. By definition, these sunflowers each contain a unique set of at least $n/c $ nodes (i.e., their kernel) and all these nodes are covered. Moreover, it is straightforward to see that at least one new root is added to the forest while there remains uncovered nodes. Hence, by the end of phase $i$, at most $n - i \cdot  n/c$ nodes can remain uncovered, with high probability. 
\end{proof}

\begin{lemma}
\label{lem:setUpGuarantee}
    Forest $F_{c^*}$ spans $G$, has depth $2T(n)$ and contains at most $c^*$ trees, with high probability.  
\end{lemma}

\begin{proof}
    By Lemma \ref{lem:setupInvariant}, forest $F_{c^*}$ has depth at most $2T(n)$. By Proposition \ref{prop:coveredNodes}, it holds with high probability that the forest $F_{c^*}$ spans $G$. Moreover, it also holds with high probability that for any two roots in the forest, their sunflowers are pairwise kernel-disjoint. Since each such kernel contains at least $n/c$ unique nodes, there can be at most $c$ roots in the forest. Since the number of roots is an integer, there are at most $c^*$ roots, and the statement follows.
\end{proof}

Next, we show that each phase increases the number of clusters each super cluster must contain. This goes on until the end of phase $\log c^*$, when a single super cluster contains all up to $c^*$ clusters obtained from the set-up part. 

\begin{lemma}
\label{lem:treeMerging}
   Let $k$ be the number of initial clusters for the tree-merging phase. With high probability, it holds that when any (tree-merging) phase $i \in \{0,\ldots,\log c^*\}$ ends, each super cluster contains at least $\min\{2^i,k\}$ clusters. 
\end{lemma}

\begin{proof}
    We prove the statement by induction on $i$. The base case for phase $i = 0$ follows trivially, since all super clusters initially contain at least 1 cluster. Next, let us assume the induction hypothesis holds for some phase $0 \leq i < \log c^*$: i.e., each super cluster contains at least $\min\{2^i,k\}$ clusters by the end of phase $i$. Then, we show that the induction step holds with high probability. (A union bound over the $\log c^* \leq n$ phases gives us the lemma statement.) Consider phase $i+1$. Super clusters cannot ``lose'' clusters (or put conversely, cluster do not ``leave'' super clusters), so any cluster that already has $\min\{2^{i+1},k\}$ clusters when phase $i+1$ starts, still contains at least as many when phase $i+1$ ends. As for the other super clusters, containing at least $\min\{2^i,k\}$ but strictly less than $\min\{2^{i+1},k\}$ clusters, then for each super cluster, its cluster with maximum root ID is within distance at most $\min\{2^{i+1},k\}$ to any other cluster, within the cluster supergraph. If $2^{i+1} < k$, then there exists at least one outgoing edge, and the super cluster samples one such outgoing edge with high probability. This outgoing edge leads to another super cluster with at least $\min\{2^i,k\}$, different clusters. This newly obtained super cluster (possibly resulting from the merging of many super clusters with $\min\{2^i,k\}$ clusters) contains at least $\min\{2^{i+1},k\}$ clusters, and the induction step follows (with some error probability polynomially small in $n$).
\end{proof}

Finally, we can prove this section's main result.

\rumorSpreadingInWeakConductance*

\begin{proof}
    By Lemma \ref{lem:setUpGuarantee}, the tree-merging phase starts with at most $k \leq c^*$ clusters with high probability, that altogether cover all nodes. Moreover, by Lemma \ref{lem:treeMerging}, any super cluster must contain at least $k$ clusters by the end of tree-merging phase $\log c^*$. In other words, a single super cluster exists by the end of tree-merging phase $\log c^*$, and it covers graph $G$. After which, we construct a tree of diameter $O(c \log n / \Phi_c)$ spanning $G$ (see Proposition \ref{lem:reorientSuperCluster} from Section~\ref{subsec:communication-primitives}).
    Finally, convergecasting and broadcasting over this tree suffices to solve rumor spreading in another $O(c \log n / \Phi_c)$ rounds.

    Let us now bound the round complexity. Each phase of the set-up part takes $O(\log n / \Phi_c)$ rounds, so over the $c^*$ phases, the set-up part takes $O(c \log n / \Phi_c)$ rounds. As for the tree-merging part, each phase $i \in \{1,\ldots,\log c^*\}$ takes $O(2^i T(n)) \leq \alpha 2^i T(n)$ rounds (for some constant $\alpha > 0$, independent of $i$). Summing up, we get that altogether, the phases take $\sum_{i = 1}^{\log c^*} \alpha 2^i T(n) = O(c^* T(n)) = O(c \log n / \Phi_c)$ rounds. Moreover, building a spanning tree of $G$ at the end of the merging part takes $O(c \log n / \Phi_c)$ rounds by Proposition \ref{lem:reorientSuperCluster}, and similarly for broadcast and convergecast over this spanning tree, by Proposition \ref{prop:treePrimitives}. Finally, it suffices to add up these round complexities.
\end{proof}

\section{Rumor Spreading for General Graphs  in $\tilde{O}(D+\sqrt{n})$ rounds}
\label{sec:diameterAlg}

A natural approach for solving rumor spreading in $\GOSSIP-\CONGEST$ within general graphs --- see e.g. \cite{GK18} --- is to sparsify the communication graph $G$. Then, one must trade off the cost of computing a sparse subgraph $G'$ of $G$ with the stretch properties of $G'$. Clearly, the latter properties impact the runtime of any rumor spreading algorithm run on $G'$, since rumor spreading will take at least the diameter of $G'$.
However, we do not have any fast sparsification methods in $\GOSSIP-\CONGEST$ to compute a sparse subgraph with low stretch over the entire graph. Instead, we can, at a reasonable runtime cost, compute a sparse subgraph $G'$ with the following stretch properties: (1) large additive stretch on shortest paths within the high degree parts of $G$, and (2) low (or even no) multiplicative stretch on the shortest paths within the low degree parts of $G$.

In this work, we extend this sparsification approach one step further, with the following intuition in mind: the runtime of simple rumor spreading on $G'$ depends not only on the stretch properties of $G'$, but also on its degree properties. More precisely, in addition to the above stretch properties, our computed sparse subgraph $G'$ ensures (informally) that on the low (but polynomial in $n$) degree parts of $G$, $G'$ also has $O(\log n)$ maximum degree.

In Section \ref{subsec:simulation}, we show how to simulate one round of $\CONGEST$ communication with a degree dependent, multiplicative blowup in $\GOSSIP-\CONGEST$. Following which, we describe our sparsification technique in Section \ref{subsec:sparseSubgraphs}. (We leverage the simulation from Section \ref{subsec:simulation} to run this sparsification efficiently.) Finally, in Section \ref{subsec:rumorSpreading}, we leverage this new sparsification technique to give an algorithm that achieves rumor spreading, for any $n$-node graph $G$ of diameter $D$, in $\tilde{O}(D + \sqrt{n})$ rounds w.h.p., in the $\GOSSIP-\CONGEST$ model.

\subsection{Vertex Cover Based Simulation of $\CONGEST$}
\label{subsec:simulation}

In general, communication within the $\GOSSIP-\CONGEST$ model is significantly more limited than that of the $\CONGEST$ model. However, when the communication graph $G$ has low maximum degree --- e.g., the maximum degree is $\Delta = O(1)$ --- then clearly the $\GOSSIP-\CONGEST$ model can efficiently simulate any algorithm designed in the $\CONGEST$ model. 
More precisely, one round of $\CONGEST$ communication can be simulated via a phase of $\Delta$ rounds of $\GOSSIP-\CONGEST$ communication. During each such phase, each node $u$ contacts all of its neighbors one at a time --- one neighbor per round and in a round robin fashion --- and the communication over an edge $e = (u,v)$ consists of $u$ and $v$ exchanging the messages they would have sent each other within the simulated $\CONGEST$ round. 

We generalize the above insight: $\CONGEST$ communication (over some subgraph) can be efficiently simulated given either (i) a ``low degree'' vertex cover or (ii) a ``low out-degree'' edge orientation --- see Lemmas \ref{lem:lowDegreeSimulation} and \ref{lem:lowOutdegreeSimulation} below.

\begin{lemma}
\label{lem:lowDegreeSimulation}
Let $G' = (V',E')$ be some (possibly disconnected) subgraph of $G$. Furthermore, let $\Deltath \geq 1$ be some integer degree threshold, and $U$ be a vertex cover of $G'$ such that the maximum degree of nodes in $U$ within $G'$ is at most $\Deltath$. Then, there exists a $\GOSSIP-\CONGEST$ algorithm simulating one round of $\CONGEST$ communication on the connected components of $G'$ using $O(\Deltath)$ rounds.
\end{lemma}

\begin{proof}
    The $\GOSSIP-\CONGEST$ algorithm simulating one round of $\CONGEST$ communication on the connected components of $G'$ works as follows. Nodes outside $U$ do not contact any neighbors, whereas nodes in $U$ contact all of their neighbors within $G'$, one at a time. Since a new neighbor can be contacted per round, this takes at most $\Deltath$ rounds. Moreover, the communication that happens when some node $v$ contacts some neighbor $u$ simply consists of an exchange of the messages that $v$ and $u$ send to each other within the simulated round of $\CONGEST$ communication.

    Finally, the correctness of the simulation follows from the fact that for every edge in $E'$, by definition of a vertex cover, at least one of its endpoints is in $U$ and that endpoint communicates over that edge (thus initiating a message exchange).
\end{proof}

\begin{lemma}
\label{lem:lowOutdegreeSimulation}
Let $G' = (V',E')$ be some (possibly disconnected) subgraph of $G$, and consider some edge orientation of $E'$ (where each edge's orientation is known to both endpoints). 
Let $\Deltath \geq 1$ be some integer degree threshold such that the maximum out-degree of $G'$ induced by the edge orientation of $E'$ is at most $\Deltath$. Then, there exists a $\GOSSIP-\CONGEST$ algorithm simulating one round of $\CONGEST$ communication on the (undirected) connected components of $G'$ using $O(\Deltath)$ rounds.
\end{lemma}

\begin{proof}
    The $\GOSSIP-\CONGEST$ algorithm simulating one round of $\CONGEST$ communication on the connected components of $G'$ works as follows. All nodes in $V'$ contact all of their outgoing neighbors within $G'$, one at a time. Since a new (outgoing) neighbor can be contacted per round, this takes at most $\Deltath$ rounds. Moreover, the communication that happens when some node $v$ contacts some neighbor $u$ simply consists of an exchange of the messages that $v$ and $u$ send to each other within the simulated round of $\CONGEST$ communication.

    Finally, the correctness of the simulation follows from the fact that for every edge in $e \in E'$, the edge orientation forces $e$ to be an outgoing edge for one of its endpoints, and that endpoint communicates over $e$.
\end{proof}

\subsection{Computing Sparse Subgraphs}
\label{subsec:sparseSubgraphs}

Now, we describe our novel sparsification technique. The sparse subgraph that we obtain satisfies the below definition --- see Definition \ref{def:detourSparseSubgraph}.

\begin{definition}
\label{def:detourSparseSubgraph}
    For any integers $\kappa, \Deltath \geq 1$, we define a \emph{$\kappa$-detour, $\Deltath$-threshold sparse subgraph} $G' = (V, E')$ of $G = (V,E)$ as follows. There exists some (distributedly known) superset $\bar{H}$ of $H$, where $H$ is the set of nodes with degree $\Omega(\Deltath)$ in $G$, such that if we let $L = V \setminus \bar{H}$ denote the remaining nodes, $E_{\bar{H}} = \{\{u,v\} \in E \mid u \in \bar{H}, v \in \bar{H} \}$, $E_L = E \setminus E_{\bar{H}}$, $E_{\bar{H}}' = E_{\bar{H}} \cap E'$ and $E_L' = E_L \cap E'$:
    \begin{enumerate}
        \item $G[E_{\bar{H}}']$ and $G[E_L']$ maintain the connectivity of, respectively, $G[E_{\bar{H}}]$ and $G[E_L]$,\footnote{For clarity, for any $F \subseteq E$, the notation $G[F]$ indicates the subgraph of $G$ whose edge set is $F$, and whose vertex set consists of all of the endpoint nodes of edges in $F$. Note that this implies, here, that the vertex set of $G[E_L]$ may contain nodes in $\bar H$, and that the vertex sets of $G[E_{\bar{H}}]$ and $G[E_L]$ may intersect.}
        \item $G[E_{\bar{H}}']$ is a spanning forest that consists of trees $T_1,\ldots,T_r$ of diameter $D_1, \ldots, D_r$, such that $\sum_{i=1}^r D_i = O(\kappa)$,
        \item $G[E_L']$ is a $O(\log n)$ spanner (forest) of $G[E_L]$,\footnote{For any graph $G$ and stretch $k \geq 1$, a $k$-spanner $G'$ of $G$ is a subgraph of $G$ such that distances within $G'$ are at least as large as those in $G$, but also at most a multiplicative factor $k$ higher. In other words, distances are stretched by at most a $k$ multiplicative factor. For a disconnected $G$, we call $G'$ a $k$-spanner forest if it defines a $k$-spanner for each connected component of $G$.}
        \item There exists some distributedly known edge orientation of $E_L'$ (i.e., each node knows for each incident edge in $E_L'$, whether it is oriented outwards or inwards) such that the maximum out-degree induced by that edge orientation is $O(\log n)$.
    \end{enumerate}  
\end{definition}

Note that if one separates the edge set $E$ of $G$ into $E_{\bar{H}}$ and $E_L$, then any shortest path in $G$ decomposes into a series of (alternating) shortest paths along $G[E_{\bar{H}}]$ and $G[E_L]$ (or vice versa, depending on the shortest paths' starting point). Let $0 \leq \delta \leq 1$ be some known constant, then the natural next step is to sparsify $E$ --- by sparsifying $E_{\bar{H}}$ into $E_{\bar{H}}'$ and $E_L$ into $E_L'$, independently --- into an edge set $E' = E_{\bar{H}}' \cup E_L'$ such that, for any shortest path in $G$, then its components along $G[E_{\bar{H}}]$ are stretched to a total of at most $O(n^{1-\delta} \log n)$ hops within $G[E_{\bar{H}}']$, and its components along $G[E_L]$ are stretched by $O(\log n)$ hops \emph{within a subgraph admitting a known, low maximum out-degree edge orientation}. This latter condition, on the low maximum (out-)degree edge orientation, is crucial to our setting; it allows us to use $\kappa$-detour, $\Deltath$-threshold sparse subgraph subgraphs to obtain efficient rumor spreading algorithms in $\GOSSIP-\CONGEST$.

Next, we describe a $\GOSSIP-\CONGEST$ algorithm for computing a $O(n^{1-\delta} \log n)$-detour, $O(n^{\delta} \log n)$-threshold sparse subgraph of $G$, for any known constant $0 \leq \delta \leq 1$. 
The algorithm takes $\tilde{O}(n^{1-\delta} + n^\delta) = \tilde{O}(n^{\max(1-\delta,\delta)})$ rounds, and consists of three types of phases: (1) a sampling phase, which computes $\bar{H}$,  
(2) several tree-merging phases, which altogether compute $E_{\bar{H}}'$, and
(3) a sparsification phase, which not only computes $E_L'$ but also an (low out-degree) edge orientation of $E_L'$.
We describe these phases below.

\paragraph{Sampling Phase.} The algorithm starts with a sampling phase that consists of $O(n^\delta \log n)$ rounds. Initially, all nodes choose to enter a \emph{star set} $S$ with probability $n^{-\delta}$. Next, all non-star nodes enter either set $H$ or set $L$ --- which stands for high degree and low degree sets respectively --- depending on whether their degree is respectively greater than or equal to some threshold $\Theta(n^\delta \log n)$, or lower than it.
Then, each star node becomes the root of a cluster. After which, each non-star high degree node $v$ takes $\Theta(n^\delta \log n)$ rounds to communicate over $\Theta(n^\delta \log n)$ randomly chosen incident edges (with replacement).
If in any such round, $v$ contacts (i.e., exchanges a message with) a star node, then after these $\Theta(n^\delta \log n)$ rounds, $v$ joins the cluster rooted at that star node (and if there are multiple, $v$ chooses one such star neighbor arbitrarily). On the other hand, each low degree node $u$ takes $\Theta(n^\delta \log n)$ rounds to communicate over all of its incident edges. This ensures that every node knows which incident edges lead to nodes in $L$. (In particular, all nodes in $G[E_L]$ know which incident edges belong to that subgraph, and all other nodes learn that they have no such incident edges, i.e., that they are not in $G[E_L]$.) 

\paragraph{Tree-Merging Phases.} Subsequently, we run $O(\log n)$ tree-merging phases on $G[S \cup H]$. (Note that $S \cup H$ corresponds to the $\bar{H}$ for the sparse subgraph being computed.) 
Initially, the clusters are exactly the clusters formed of a star node and the high degree neighbors having joined that cluster. (Moreover, recall that each node knows, for each incident edge, whether the other endpoint is in $L$.) 
These phases are used to merge these clusters in a controlled manner until all nodes eventually belong to only one cluster.

Each phase takes $O(n^{1-\delta} \log n)$ rounds.
At the start, each cluster flips a coin to choose ``heads'' or ``tails'' with equal probability. If the coin was ``tails'', the cluster then samples an outgoing edge in $G[S \cup H]$ (see Section~\ref{subsec:communication-primitives}, where nodes computes graph sketches only on edges with no endpoints in $L$), and otherwise the cluster waits for some incoming communication. Any sampled outgoing edge is kept only if the other endpoint's cluster chose ``heads''. Finally, the clusters connected together by the (kept) edges are merged together in a larger cluster. Note that the randomized process ensures that if we consider each larger cluster as a (cluster) supergraph, where clusters are the nodes and two clusters are connected by an edge if there is some kept sampled edge between them, then the supergraph has diameter at most $2$. 

After $O(\log n)$ such tree-merging phases, every connected component of $G[S \cup H]$ is spanned by a single tree, with diameter $O(n^{1-\delta} \log n)$, w.h.p. 
Then, each node in $G[S \cup H]$ adds its incident edges within that spanning tree, to the edge set $E_{\bar{H}}' \subseteq E'$ of the output sparse subgraph.

\paragraph{Sparsification Phase.} Finally, we run a sparsification phase on $G[E_L] = (U, E_L) $, where $E_L$ denotes all of the edges incident to any node in $L$. This sparsification phase takes $O(n^\delta \log^2 n)$ rounds and it consists of the following steps. First, we run a $\GOSSIP-\CONGEST$ simulation of the MPX algorithm \cite{MPX13} on $G[E_L]$ using $O(n^\delta \log^2 n)$ rounds --- $L$ is a ``low degree'' vertex cover of $G[E_L]$, thus we can use the simulation from Lemma \ref{lem:lowDegreeSimulation}. By simulating the MPX algorithm, we partition $U$ into clusters of diameter $O(\log n)$ w.h.p., and such that for each node $v \in U$, the neighbors of $v$ are contained within at most $O(\log n)$ clusters w.h.p. Subsequently, after the simulation of MPX is complete, each MPX cluster's root broadcasts its ID to all of the cluster nodes in $O(\log n)$ rounds (see Proposition \ref{prop:treePrimitives} in Section \ref{subsec:communication-primitives}). After which, nodes in $L$ simulate one round of $\CONGEST$ communication over $G[E_L]$, using $O(n^\delta \log^2 n)$ rounds, so that all nodes of $G[E_L]$ learn the cluster's ID of each of their neighbors (notice that for a given node, several of its neighbors might belong to the same cluster). Then, each node in $G[E_L]$ keeps a single edge per neighboring cluster, as well as the edge leading to its parent (if it has one) within its MPX cluster, and adds these edges to the edge set $E_L' \subseteq E'$ of the output sparse subgraph. Moreover, each node initially orients these edges outwards, but nodes in $L$ simulate a final round of $\CONGEST$ communication over $G[E_L]$ so that (i) the edge orientation is known to both endpoints, and (ii) if both endpoints initially oriented the same edge outwards (which would result in an invalid edge orientation), then the edge is oriented outwards from the endpoint with largest ID.

\paragraph{Analysis.} We start by giving two lemmas for the sampling phase.

\begin{lemma}
\label{lem:concentrationOfStars}
    With high probability, it holds that:
    \begin{enumerate}
        \item There are at most $O(n^{1-\delta} \log n)$ star nodes, w.h.p.
        \item Any high-degree node has a star neighbor, w.h.p. More precisely, any high-degree node $v$ has $\Theta(d(v) n^{-\delta})$ star neighbors w.h.p.
    \end{enumerate}
\end{lemma}

\begin{proof}
    Recall that each node becomes a star node with probability $n^{-\delta}$, independently and uniformly at random. Hence, out of all nodes, $O(n^{1-\delta})$ become star nodes in expectation.
    By a simple application of the Chernoff bounds in Lemma \ref{lem:ChernoffBound}, we see that $O(n^{1-\delta} \log n)$ nodes becomes stars w.h.p., thus proving the first item. As for the second item, note that any high degree node $v$, i.e., with degree $d(v) = \Omega(n^{\delta} \log n)$, has in expectation $d(v) n^{-\delta} = \Omega(\log n)$ stars as neighbors. Once again, by a simple application of the Chernoff bounds in Lemma \ref{lem:ChernoffBound}, we see that any high degree node has $\Theta(d(v) n^{-\delta}) = \Omega(\log n)$  star neighbors w.h.p., thus proving the second item.
\end{proof}

\begin{lemma}
\label{lem:starClusters}
    Within the $\Theta(n^{\delta} \log n)$ rounds of the sampling phase, any non-star high degree node chooses an edge leading to a star neighbor at least once, w.h.p.
\end{lemma}

\begin{proof} 
    By Lemma \ref{lem:concentrationOfStars}, any non-star high degree node $v$ has $\Theta(d(v) n^{-\delta})$ star neighbors w.h.p. Now, in each round of the sampling phase, node $v$ contacts a neighbor chosen independently and uniformly at random. Hence, in each round node $v$ contacts a star neighbor with probability $\frac{\Theta(d(v) n^{-\delta})}{d(v)} = \Theta(n^{-\delta})$. Hence, over $\Theta(n^{\delta} \log n)$ rounds, $v$ contacts in expectation $\Theta(\log n)$ star neighbors. Finally, by a simple application of the Chernoff bounds in Lemma \ref{lem:ChernoffBound}, we see that during the sampling phase, $v$ contacts at least one star neighbor w.h.p.
\end{proof}

Next, we prove that the tree-merging phases indeed compute trees spanning each connected component of $G[S \cup H]$, and whose diameter is linear in the number of (star) clusters present initially in the component.

\begin{lemma}
\label{lem:treeMergingStars}
    Recall that $E_{\bar{H}}'$ are the edges computed throughout the $O(\log n)$ tree-merging phases. Overall, these tree-merging phases take $O(n^{1-\delta} \log^2 n)$ rounds, and ensure that w.h.p., each connected component of $G[S \cup H]$ induces a spanning tree in $G[E_{\bar{H}}']$, of diameter $O(X)$, where $X$ is the number of stars in that component.
\end{lemma}

\begin{proof}
    Consider any one connected component $\mathcal{C}$ of $G[S \cup H]$, and let $X$ be the number of stars within that component. By Lemma \ref{lem:concentrationOfStars}, $X = O(n^{1-\delta} \log n)$ w.h.p.
    Next, we show by induction that for any tree-merging phase $i \geq 0$, the diameter of any cluster in $\mathcal{C}$ when the phase ends is at most $\min\{5^{i+1}, 3X\}$. The base case is simple, since initially all clusters have diameter at most 2 (and if $X = 0$, there are no clusters). Next, assume the induction hypothesis holds for some $i \geq 0$. Then, in phase $i+1$, the clusters' diameter can grow to at most $3*5^i + 2 \leq 5^{i+1}$ 
    (since the diameter of the cluster supergraph is at most 2.)
    Moreover, it is easy to see that for each cluster, there are at most $X$ stars contained within, and the stars are a dominating set on the cluster's spanning tree.  Hence, the diameter of any cluster in $\mathcal{C}$ is at most  $\min\{5^{i+1}, 3X\}$. 
    This proves the induction step. Note that this inductive statement proves that each tree-merging phase can be executed in this model within $O(n^{1-\delta} \log n)$ rounds.

    Finally, we show for each tree-merging phase $i \geq 0$, if we condition on the phase starting with some $n_i \geq 2$ clusters, then the phase ends with at most $3n_i/4$ clusters in expectation. First, note since all nodes know which incident edges lead to $L$, then all outgoing sampled edges are in $G[S \cup H]$. Thus, any cluster $C$ merges into another cluster if it flips tails, and looking at its outgoing edge, the cluster containing the other endpoint flipped heads. 
    Hence, the probability that any other combination (of heads and tails for the two clusters) happens, which is at most $3/4$, is an upper bound on the probability that $C$ does not merge with any cluster. The statement follows, and using the law of total expectation, we can show that if there are initially $X$ clusters in component $\mathcal{C}$, then by the end of phase $i \geq 1$, the expected number of clusters remaining is $X\cdot(3/4)^i \leq n \cdot (3/4)^i$. Finally, it suffices to use Markov's inequality to show that after some $O(\log n)$ tree-merging phases, taking altogether $O(n^{1-\delta} \log^2 n)$ rounds, a single cluster $C^*$ remains with high probability, and component $\mathcal{C}$ is spanned by the cluster tree of $C^*$.
\end{proof}

Finally, we show that the sparsification phase not only reduces the edge set of each connected component of $G[E_L]$ down to an $O(\log n)$ stretch spanner, but also computes an $O(\log n)$ maximum out-degree edge orientation for that spanner.

\begin{lemma}
\label{lem:sparsificationPhase}
    Recall that $E_L'$ are the edges computed in the sparsification phase. The sparsification phase takes $O(n^{\delta} \log^2 n)$ rounds, and ensures that $G[E_L']$ is an $O(\log n)$-spanner (forest) of $G[E_L]$. Furthermore, it ensures that maximum out-degree of $G[E_L']$ induced by the edge orientation computed in the sparsification phase is $O(\log n)$. 
\end{lemma}

\begin{proof}
First, we point out that executing the MPX algorithm of \cite{MPX13} on some $n$-node graph $G = (V,E)$ takes $O(\log n)$ rounds, and outputs a low diameter graph decomposition of $G$: that is, a partition of the vertex set $V$ into subsets $V_1,\ldots,V_k$ such that each cluster $V_i$ has strong diameter at most $O(\log n)$ and is spanned by a tree of depth $O(\log n)$.
Since, by definition, $L$ is a vertex cover of $E_L$ and nodes in $L$ have at most $O(n^\delta \log n)$ incident edges in $G$, then by Lemma \ref{lem:lowDegreeSimulation}, the $O(\log n)$ rounds of the MPX algorithm on (the connected components of) $G[E_L]$ can indeed be simulated by a $O(n^\delta \log^2 n)$ round $\GOSSIP-\CONGEST$ algorithm. Furthermore, at the end of the MPX algorithm, each node with a parent in its MPX cluster's tree knows which incident edge leads to that parent.

After computing such a low diameter graph decomposition, broadcasting the ID of each cluster's root over each cluster's spanning tree takes $O(\log n)$ rounds by Proposition \ref{prop:treePrimitives}. Then, by Lemma \ref{lem:lowDegreeSimulation}, a single round of communication over (the connected components of) $G[E_L]$ can indeed be simulated by a $O(n^\delta \log n)$ round $\GOSSIP-\CONGEST$ algorithm. At which point, each node in $G[E_L]$ knows, for each incident edge $e \in E_L$, whether the other endpoint of edge $e$ lies in the same cluster. 

Thus, by the algorithm description, each node in $G[E_L]$ adds one incident edge to $E_L'$ per neighboring cluster (as well as the edge leading to its parent in the MPX cluster). Now, note that Corollary 3.9 of \cite{HW16} proves that with high probability, each node $v \in V$ neighbors $O(\log n)$ clusters, or in other words, the neighbors of $v$ belong to at most $O(\log n)$ clusters. Hence, by orienting these edges (and the edge leading to the parent in the MPX cluster) outwards from $v$, the resulting edge orientation induces a maximum out-degree in $G[E_L']$ of $O(\log n)$. (If two endpoints of an edge both attempt to orient it outwards, the nodes break the tie and choose the direction arbitrarily, for example, by directing it outwards from the highest ID. Notice that this does not increase the maximum out-degree but may only decrease it.)
As for the $O(\log n)$ stretch, consider any shortest path from any two nodes $u,v$ within one connected component of $G[E_L]$. (Note that we ignore the edge orientation here.) Then, any edge in that path is replaced in $G[E_L']$ by a path of at most $O(\log n)$ hops that includes at most one inter-cluster edge, and $O(\log n)$ edges within a cluster's spanning tree. 
\end{proof}

Finally, we show that our sparsification technique indeed computes an $O(n^{1-\delta} \log n)$-detour, $O(n^{\delta} \log n)$-threshold sparse subgraph of $G$.

\begin{theorem}
\label{thm:sparsificationGossip}
    For any constant $0 \leq \delta  \leq 1$, there exists a $\GOSSIP-\CONGEST$ algorithm that computes a $O(n^{1-\delta} \log n)$-detour, $O(n^{\delta} \log n)$-threshold sparse subgraph of $G$, in $O(n^{\max(1-\delta,\delta)} \log^2 n)$ rounds.
\end{theorem}

\begin{proof}
Let $\bar{H} = S \cup H$, then $\bar{H}$ is a superset of $H$, the nodes with degree at least $\Omega(n^{\delta} \log n)$ in $G$. Moreover, by the algorithm description, $G' = (V, E')$ where $E' = E_{\bar{H}'} \cup E_L'$. Next, consider $G[E_{\bar{H}}']$. By Lemma \ref{lem:treeMergingStars}, each connected component $G[\bar{H}] = G[E_{\bar{H}}]$ induces a spanning tree in $G[E_{\bar{H}}']$, thus maintains the connectivity of $G[E_{\bar{H}}]$. Moreover, each such spanning tree $T_i$ has diameter $D_i = O(X_i)$, where $X_i$ is the number of stars in $T_i$. Since there are at most $O(n^{1-\delta} \log n)$ star nodes w.h.p., by Lemma \ref{lem:concentrationOfStars}, then it follows that $\sum_{i=1}^r D_i = O(n^{1-\delta} \log n)$.
Finally, consider $G[E_L']$.  Then, by Lemma \ref{lem:sparsificationPhase}, $G[E_L']$ is an $O(\log n)$-spanner (forest) of $G[E_L]$. In particular, this implies that $G[E_L']$ maintains the connectivity of $G[E_L]$. Furthermore, Lemma \ref{lem:sparsificationPhase} shows that the edge orientation computed during the sparsification phase (and which is distributedly known) induces a maximum out-degree of $O(\log n)$ on $G[E_L]$. Hence, $G'$ is indeed a $O(n^{1-\delta} \log n)$-detour, $O(n^{\delta} \log n)$-threshold sparse subgraph of $G$.

\begin{sloppypar}
Finally, the runtime for computing the sparse subgraph $G'$ is $O(n^\delta \log n) + O(n^{1-\delta} \log^2 n) + O(n^{\delta} \log^2 n) = O(n^{\max(1-\delta,\delta)} \log^2 n)$.
\end{sloppypar}
\end{proof}

\subsection{Rumor Spreading}
\label{subsec:rumorSpreading}

We now describe our rumor spreading algorithm. We start by computing a $O(\sqrt{n} \log n)$-detour, $O(\sqrt{n} \log n)$-threshold sparse subgraph $G' = (V,E')$ in $\tilde{O}(\sqrt n)$ rounds --- see Theorem \ref{thm:sparsificationGossip}. Let $\bar{H}$ denote the (distributedly known) node superset of the nodes with degree $\Omega(\sqrt{n} \log n)$ in $G$, and 
$E_{\bar{H}}'$ and $E_L'$ the two edge subsets of $E'$ as described in Definition \ref{def:detourSparseSubgraph}. 

Next, all nodes participate in $O((D + \sqrt n) \log n)$ multiple rumor spreading phases, each of which consists of $O(\log n)$ rounds, and terminate. Each phase ensures that the rumor spreads by at least one hop along any shortest path in $G'$. 
More precisely, during each phase, the first round consists of each node in $\bar{H}$ contacting their parent within $E_H'$ --- recall that $G[E_{\bar{H}}']$ is a spanning forest --- and both nodes exchange the rumor, if they have it, along that edge. 
The remainder of the phase's rounds are used to simulate one round of $\CONGEST$ communication between the nodes in $G[E_L']$. Note that such a round can indeed be simulated with a $O(\log n)$ blow up in $\GOSSIP-\CONGEST$, as nodes (distributedly) know an edge orientation of $G[E_L']$ inducing a maximum out-degree of $O(\log n)$ --- see Lemma \ref{lem:lowOutdegreeSimulation}.

\rumorSpreadingGeneral*

\begin{proof} 
\begin{sloppypar}
By Theorem \ref{thm:sparsificationGossip}, computing a $O(\sqrt{n} \log n)$-detour, $O(\sqrt{n} \log n)$-threshold sparse subgraph $G' = (V,E')$ takes $O(\sqrt n \log^2 n)$ rounds. Then, it remains to show that the rumor is spread, from some source node $u$ to all nodes in $G$ via $G'$, in $\tilde{O}(D+\sqrt n)$ rounds.
\end{sloppypar}

Next, let $u$ be the source that initially holds the rumor. Then, by the algorithm description and the properties of $G'$, we can show that the rumor spreads by one hop within $G'$ per phase. Indeed, the first round of a phase suffices to spread the rumor by one hop along $G[E_{\bar{H}}]$, since any node $v$ holding the rumor within $G[E_{\bar{H}}]$ lies in a tree (by definition of $G'$), contacts its parent and is contacted by its children, and thus by the algorithm description, exchanges the rumor with its parent and its children within that round. As for spreading the rumor by one hop along $G[E_L']$, note that the distributedly known edge orientation of $G[E_L']$ induces a maximum out-degree of $O(\log n)$, thus by Lemma \ref{lem:lowOutdegreeSimulation}, $O(\log n)$ rounds of $\GOSSIP-\CONGEST$ can simulate a single round of $\CONGEST$ communication in which each node within $G[E_L']$ with the rumor sends it to all of its neighbors (in $G[E_L']$). To sum up, a rumor spreading phase succeeds in spreading the rumor by (at least) one hop along $G[E']$.

Finally, it remains to consider the stretch induced by $G'$ compared to $G$. On $G[E_L]$, $G'$ induces a multiplicative, $O(\log n)$ stretch. On $G[E_{\bar{H}}]$, $G'$ induces an additive, $O(\sqrt n \log n)$ stretch (over the sum of all spanning trees). Since the diameter $D$ upper bounds the distance from $u$ to any other node $v \in V$ within $G$, we get that the rumor spreading in $G'$ succeeds in $O((D +\sqrt n) \log n)$ rumor spreading phases, and thus in  $O((D +\sqrt n) \log^2 n)$ rounds.
\end{proof}

\paragraph{Removing Knowledge of $D$.} Note that the rumor spreading algorithm's description given above, and the obtained runtime, relies on all nodes knowing constant factor upper bounds on both the diameter $D$ and the number of nodes $n$. However, it suffices to know a constant factor upper bound on $n$ only, since we can estimate the diameter $D$ in an exponentially growing fashion. 

To do so, modify the above described rumor spreading algorithm by having nodes start with the same, but constant, estimate of $D$ --- denoted by $D_{est}$. Then, nodes run the above algorithm for $\tilde{O}(D_{est}+\sqrt n)$ rounds, and keep track of which neighbor they first received the rumor from. This forms a tree, rooted in the node with the rumor, and the tree spans $G$ if and only if the rumor has spread to every node. This provides a simple condition with which nodes can detect whether the rumor spreading algorithm has terminated. We describe the detection mechanism.

Every node that received the rumor, must be part of some tree. These nodes detect termination as follows. 
First, they aggregate information up the tree, so that the root can sample an outgoing edge w.h.p., if there exists one (see Section~\ref{subsec:communication-primitives}). Following which, the root executes a broadcast. Depending on whether the sampling procedure succeeds or not --- where if the procedure fails, then this implies that there are no outgoing edges --- the broadcast informs all nodes in the tree to, respectively, double their estimate $D_{est}$ or to terminate. 
Note that all nodes without the rumor, and by extension not within the tree, know that the estimate $D_{est}$ is incorrect and thus double it. Since all nodes know the same constant approximation of $n$ and the same estimate $D_{est}$, they can wait until the correct round to start the next iteration of the algorithm with the new estimate.

\section{Applications}
\label{sec:apps}

We  outline how our gossip algorithms can be used to solve other fundamental problems, such as computing aggregate functions, MST, and leader election. 

As mentioned earlier, both our gossip algorithms construct spanning trees over the graph as a byproduct, where the diameter of the spanning tree is upper bounded by the run time of the algorithm. Such a spanning tree
can be constructed as follows: each node remembers the first edge along which it receives the rumor, and denotes the corresponding node as its parent in the tree. 

On such spanning trees, leader election is trivial, as the root can be chosen as the leader and this information subsequently broadcasted to all nodes. 

Additionally, for functions where one may aggregate values, the root may learn the aggregate function over the entire graph in a single convergecast and subsequently spread this information to all nodes via a single broadcast. Thus, we see the following two sets of results. For graphs with weak conductance $\Phi_c$, we see that spanning tree construction, leader election, and (exact) aggregate function computation may be performed with high probability in $O(c \log n / \Phi_c)$ rounds. For general graphs, we see that spanning tree construction, leader election, and aggregate function computation may be performed with high probability in $\tilde{O}(D + \sqrt{n})$ rounds.

Additionally, for general graphs, we can also construct an MST in $\tilde{O}(D + \sqrt{n})$ rounds. Recall that a byproduct of the algorithm for general graphs is a spanning tree of diameter $\tilde{O}(D + \sqrt{n})$. We can use this as a communication backbone, similar to the MST algorithm of Gmyr and Pandurangan~\cite{GP18}.  In short, the MST construction is separated into two stages. Each stage is decomposed into $\Theta(\log n)$ phases, each of which merges MST fragments. The first phase merges fragments up to a diameter of $O(\sqrt n)$ --- any fragment with a higher diameter stops participating --- using the fragments' spanning trees for communication (i.e., to aggregate sketches at the root, so the root can sample a minimum-weight outgoing edge). At the end of the first stage, it holds that $O(\sqrt{n})$ fragments remain, each of diameter $O(\sqrt{n})$. The second stage proceeds through similar fragment-merging phases, in which the fragments are ``virtually'' merged (i.e., their fragment IDs are unified) while retaining their individual fragment spanning trees (as in~\cite{GP18}). However, at this stage, communication (to sample the minimum-weight outgoing edge by aggregating sketches of edges from the same fragments) is performed over the communication backbone.
Both stages can be implemented gossip-efficiently, and thus, an MST is constructed in $\tilde{O}(D + \sqrt{n})$ rounds.

\section{Conclusion}
\label{sec:conclusion}

We presented the first-known gossip-based rumor spreading algorithm
that uses small-sized messages and achieves a run time that depends
on the weak conductance of the graph. Our algorithm is at least as good as the standard uniform gossip algorithm and, in many graphs, can be significantly (even exponentially, e.g., on dumbbell graphs) faster. 
Our algorithm improves on that of~\cite{CS11,CS12}, which used large message sizes and was also slower. Our algorithm's performance bound in terms of weak conductance is essentially asymptotically optimal.

We also presented a gossip-based rumor spreading algorithm
that uses small-sized messages whose run time is independent of conductance (or weak conductance) and depends
on the diameter of the graph.  Our algorithm improves upon the previous best-known algorithm of~\cite{GK18}, which also uses small messages.

The most important open question is whether we can reduce our gossip time bound of $\tilde{O}(D+\sqrt{n})$ for {\em rumor spreading} in general graphs. {\em We conjecture that
this bound is essentially tight, i.e., 
$\tilde{\Omega}(D+\sqrt{n})$ is a round lower bound
for gossip-based rumor spreading.} We note that our  algorithm can also be used to construct an MST in $\tilde{O}(D+\sqrt{n})$
gossip rounds, which is the optimal bound (up to logarithmic factors) one can obtain in terms of $D$ and $n$, since
$\tilde{\Omega}(D+\sqrt{n})$ is a well-established lower bound for MST even in the standard (non-gossip) $\CONGEST$ model~\cite{PelegR00,stoc11}.

\begin{acks}
Gopal Pandurangan was supported in part by Army Research Office (ARO) grant W911NF-231-0191 and National Science Foundation (NSF) grant CCF-2402837.
\end{acks}

\bibliographystyle{ACM-Reference-Format}
\bibliography{references}

@article{PelegR00,
  author       = {David Peleg and
                  Vitaly Rubinovich},
  title        = {A Near-Tight Lower Bound on the Time Complexity of Distributed Minimum-Weight
                  Spanning Tree Construction},
  journal      = {{SIAM} J. Comput.},
  volume       = {30},
  number       = {5},
  pages        = {1427--1442},
  year         = {2000}
}

@article{stoc11,
  author       = {Atish Das Sarma and
                  Stephan Holzer and
                  Liah Kor and
                  Amos Korman and
                  Danupon Nanongkai and
                  Gopal Pandurangan and
                  David Peleg and
                  Roger Wattenhofer},
  title        = {Distributed Verification and Hardness of Distributed Approximation},
  journal      = {{SIAM} J. Comput.},
  volume       = {41},
  number       = {5},
  pages        = {1235--1265},
  year         = {2012}
}

@InProceedings{Giakkoupis2011,
  author =	{George Giakkoupis},
  title =	{{Tight bounds for rumor spreading in graphs of a given conductance}},
  booktitle =	{28th International Symposium on Theoretical Aspects of Computer Science (STACS 2011) },
  pages =	{57--68},
  year =	{2011},
}

@inproceedings{CS11,
  title={Fast Information Spreading in Graphs with Large Weak Conductance},
  author={Censor-Hillel, Keren and Shachnai, Hadas},
  booktitle={Proceedings of the Twenty-Second Annual ACM-SIAM Symposium on Discrete Algorithms},
  pages={440--448},
  year={2011},
  organization={SIAM}
}

@article{CS12,
author = {Censor-Hillel, Keren and Shachnai, Hadas},
title = {Fast Information Spreading in Graphs with Large Weak Conductance},
journal = {SIAM Journal on Computing},
volume = {41},
number = {6},
pages = {1451-1465},
year = {2012},
doi = {10.1137/110845380},
URL = { 
        https://doi.org/10.1137/110845380
}
}

@inproceedings{JST11,
  author    = {Hossein Jowhari and
               Mert Saglam and
               G{\'{a}}bor Tardos},
  editor    = {Maurizio Lenzerini and
               Thomas Schwentick},
  title     = {Tight bounds for Lp samplers, finding duplicates in streams, and related
               problems},
  booktitle = {Proceedings of the 30th {ACM} {SIGMOD-SIGACT-SIGART} Symposium on
               Principles of Database Systems, {PODS} 2011, June 12-16, 2011, Athens,
               Greece},
  pages     = {49--58},
  publisher = {{ACM}},
  year      = {2011},
  url       = {https://doi.org/10.1145/1989284.1989289},
  doi       = {10.1145/1989284.1989289},
  bibsource = {dblp computer science bibliography, https://dblp.org}
}

@inproceedings{AhnGM12a,
  author = {Kook Jin Ahn and Sudipto Guha and Andrew McGregor},
  title = {Analyzing graph structure via linear measurements},
  booktitle = {Proceedings of the 23rd Annual {ACM-SIAM} Symposium on Discrete Algorithms (SODA)},
  pages = {459--467},
  year = {2012}
}

@article{PRS18,
  title={Fast distributed algorithms for connectivity and MST in large graphs},
  author={Pandurangan, Gopal and Robinson, Peter and Scquizzato, Michele},
  journal={ACM Transactions on Parallel Computing (TOPC)},
  volume={5},
  number={1},
  pages={1--22},
  year={2018},
  publisher={ACM New York, NY, USA}
}

@book(Upfalbook,
 title={Probability and computing: randomization and probabilistic techniques in algorithms and data analysis},
  author={Mitzenmacher, Michael and Upfal, Eli},
  year={2017},
  publisher={Cambridge university press}
)

@article{CGLP18,
  title={Rumor spreading and conductance},
  author={Chierichetti, Flavio and Giakkoupis, George and Lattanzi, Silvio and Panconesi, Alessandro},
  journal={Journal of the ACM (JACM)},
  volume={65},
  number={4},
  pages={1--21},
  year={2018},
  publisher={ACM New York, NY, USA}
}

@inproceedings{CHKM12,
  title={Global computation in a poorly connected world: fast rumor spreading with no dependence on conductance},
  author={Censor-Hillel, Keren and Haeupler, Bernhard and Kelner, Jonathan and Maymounkov, Petar},
  booktitle={Proceedings of the forty-fourth annual ACM symposium on Theory of computing},
  pages={961--970},
  year={2012}
}

@inproceedings{gia2,
  author       = {George Giakkoupis and
                  Philipp Woelfel},
  editor       = {Dana Randall},
  title        = {On the Randomness Requirements of Rumor Spreading},
  booktitle    = {Proceedings of the Twenty-Second Annual {ACM-SIAM} Symposium on Discrete
                  Algorithms, {SODA} 2011, San Francisco, California, USA, January 23-25,
                  2011},
  pages        = {449--461},
  publisher    = {{SIAM}},
  year         = {2011},
  url          = {https://doi.org/10.1137/1.9781611973082.36},
  doi          = {10.1137/1.9781611973082.36},
  bibsource    = {dblp computer science bibliography, https://dblp.org}
}

@inproceedings{vertex-exp,
  author       = {George Giakkoupis and
                  Thomas Sauerwald},
  editor       = {Yuval Rabani},
  title        = {Rumor spreading and vertex expansion},
  booktitle    = {Proceedings of the Twenty-Third Annual {ACM-SIAM} Symposium on Discrete
                  Algorithms, {SODA} 2012, Kyoto, Japan, January 17-19, 2012},
  pages        = {1623--1641},
  publisher    = {{SIAM}},
  year         = {2012},
  url          = {https://doi.org/10.1137/1.9781611973099.129},
  doi          = {10.1137/1.9781611973099.129},
  bibsource    = {dblp computer science bibliography, https://dblp.org}
}

@inproceedings{gia1,
  author       = {George Giakkoupis and
                  Frederik Mallmann{-}Trenn and
                  Hayk Saribekyan},
  editor       = {Peter Robinson and
                  Faith Ellen},
  title        = {How to Spread a Rumor: Call Your Neighbors or Take a Walk?},
  booktitle    = {Proceedings of the 2019 {ACM} Symposium on Principles of Distributed
                  Computing, {PODC} 2019, Toronto, ON, Canada, July 29 - August 2, 2019},
  pages        = {24--33},
  publisher    = {{ACM}},
  year         = {2019},
  url          = {https://doi.org/10.1145/3293611.3331622},
  doi          = {10.1145/3293611.3331622},
  bibsource    = {dblp computer science bibliography, https://dblp.org}
}

@article{CHKM12-journal,
  author       = {Keren Censor{-}Hillel and
                  Bernhard Haeupler and
                  Jonathan A. Kelner and
                  Petar Maymounkov},
  title        = {Rumor Spreading with No Dependence on Conductance},
  journal      = {{SIAM} J. Comput.},
  volume       = {46},
  number       = {1},
  pages        = {58--79},
  year         = {2017},
  url          = {https://doi.org/10.1137/14099992X},
  doi          = {10.1137/14099992X},
  bibsource    = {dblp computer science bibliography, https://dblp.org}
}

@inproceedings{H13,
  title={Simple, fast and deterministic gossip and rumor spreading},
  author={Haeupler, Bernhard},
  booktitle={Proceedings of the twenty-fourth annual ACM-SIAM symposium on Discrete algorithms},
  pages={705--716},
  year={2013}
}

@article{H15,
  title={Simple, fast and deterministic gossip and rumor spreading},
  author={Haeupler, Bernhard},
  journal={Journal of the ACM (JACM)},
  volume={62},
  number={6},
  pages={1--18},
  year={2015},
  publisher={ACM New York, NY, USA}
}

@article{FPRU90,
  title={Randomized broadcast in networks},
  author={Feige, Uriel and Peleg, David and Raghavan, Prabhakar and Upfal, Eli},
  journal={Random Structures \& Algorithms},
  volume={1},
  number={4},
  pages={447--460},
  year={1990},
  publisher={Wiley Online Library}
}

@article{FG85,
  title={The shortest-path problem for graphs with random arc-lengths},
  author={Frieze, Alan M. and Grimmett, Geoffrey R.},
  journal={Discrete Applied Mathematics},
  volume={10},
  number={1},
  pages={57--77},
  year={1985},
  publisher={Elsevier}
}

@article{P87,
  title={On spreading a rumor},
  author={Pittel, Boris},
  journal={SIAM Journal on Applied Mathematics},
  volume={47},
  number={1},
  pages={213--223},
  year={1987},
  publisher={SIAM}
}

@inproceedings{KSSV00,
  title={Randomized rumor spreading},
  author={Karp, Richard and Schindelhauer, Christian and Shenker, Scott and Vocking, Berthold},
  booktitle={Proceedings 41st Annual Symposium on Foundations of Computer Science},
  pages={565--574},
  year={2000},
  organization={IEEE}
}

@inproceedings{CLP10-soda,
  title={Rumour spreading and graph conductance},
  author={Chierichetti, Flavio and Lattanzi, Silvio and Panconesi, Alessandro},
  booktitle={Proceedings of the twenty-first annual ACM-SIAM symposium on Discrete Algorithms},
  pages={1657--1663},
  year={2010},
  organization={SIAM}
}

@inproceedings{CLP10-stoc,
  title={Almost tight bounds for rumour spreading with conductance},
  author={Chierichetti, Flavio and Lattanzi, Silvio and Panconesi, Alessandro},
  booktitle={Proceedings of the forty-second ACM symposium on Theory of computing},
  pages={399--408},
  year={2010}
}

@article{CLP11,
  title={Rumor spreading in social networks},
  author={Chierichetti, Flavio and Lattanzi, Silvio and Panconesi, Alessandro},
  journal={Theoretical Computer Science},
  volume={412},
  number={24},
  pages={2602--2610},
  year={2011},
  publisher={Elsevier}
}

@inproceedings{GSS14,
  title={Randomized rumor spreading in dynamic graphs},
  author={Giakkoupis, George and Sauerwald, Thomas and Stauffer, Alexandre},
  booktitle={International Colloquium on Automata, Languages, and Programming},
  pages={495--507},
  year={2014},
  organization={Springer}
}

@inproceedings{E06,
  title={On the communication complexity of randomized broadcasting in random-like graphs},
  author={Els{\"a}sser, Robert},
  booktitle={Proceedings of the eighteenth annual ACM symposium on Parallelism in algorithms and architectures},
  pages={148--157},
  year={2006}
}

@inproceedings{GK18,
  title={Distributed MST and broadcast with fewer messages, and faster gossiping},
  author={Ghaffari, Mohsen and Kuhn, Fabian},
  booktitle={32nd International Symposium on Distributed Computing (DISC 2018)},
  volume={121},
  pages={30--1},
  year={2018},
  organization={Schloss Dagstuhl-Leibniz-Zentrum fuer Informatik}
}

@inproceedings{GP18,
  title={Time-Message Trade-Offs in Distributed Algorithms},
  author={Gmyr, Robert and Pandurangan, Gopal},
  booktitle={32nd International Symposium on Distributed Computing (DISC 2018)},
  pages={32--1},
  year={2018},
  organization={Schloss Dagstuhl--Leibniz-Zentrum f{\"u}r Informatik}
}

@article{ER60,
author = {Erdös, Paul and Rado, Richard},
title = {Intersection Theorems for Systems of Sets},
journal = {Journal of the London Mathematical Society},
volume = {s1-35},
number = {1},
pages = {85-90},
doi = {https://doi.org/10.1112/jlms/s1-35.1.85},
url = {https://londmathsoc.onlinelibrary.wiley.com/doi/abs/10.1112/jlms/s1-35.1.85},
year = {1960}
}

@inproceedings{MPX13,
author = {Miller, Gary L. and Peng, Richard and Xu, Shen Chen},
title = {Parallel Graph Decompositions Using Random Shifts},
year = {2013},
isbn = {9781450315722},
publisher = {Association for Computing Machinery},
address = {New York, NY, USA},
url = {https://doi.org/10.1145/2486159.2486180},
doi = {10.1145/2486159.2486180},
booktitle = {Proceedings of the Twenty-Fifth Annual ACM Symposium on Parallelism in Algorithms and Architectures},
pages = {196–203},
numpages = {8},
location = {Montr\'{e}al, Qu\'{e}bec, Canada},
series = {SPAA '13}
}

@inproceedings{HW16,
author = {Haeupler, Bernhard and Wajc, David},
title = {A Faster Distributed Radio Broadcast Primitive: Extended Abstract},
year = {2016},
isbn = {9781450339643},
publisher = {Association for Computing Machinery},
address = {New York, NY, USA},
url = {https://doi.org/10.1145/2933057.2933121},
doi = {10.1145/2933057.2933121},
booktitle = {Proceedings of the 2016 ACM Symposium on Principles of Distributed Computing},
pages = {361–370},
numpages = {10},
location = {Chicago, Illinois, USA},
series = {PODC '16}
}

@article{DMP26arxiv,
  author    = {Dufoulon, Fabien  and
                  Moses Jr., William K. and
                  Pandurangan, Gopal 
               },
  title     = {Fast Gossip-Based Rumor Spreading Using Small Messages},
  journal   = {CoRR},
  volume    = {abs/2605.14376},
  year      = {2026},
  url       = {https://arxiv.org/abs/2605.14376},
  archivePrefix = {arXiv},
  eprint    = {2605.14376}
}

\balance
\end{document}